\documentclass{article}
\usepackage{amsmath,amsfonts,amsthm,color}
\usepackage{color}
\usepackage{graphicx}
\usepackage{hyperref}
\newtheorem{assumption}{Assumption}
\newtheorem{corollary}{Corollary}
\newtheorem{proposition}{Proposition}
\newtheorem{lemma}{Lemma}

\theoremstyle{remark}
\newtheorem{remark}{Remark}
\usepackage[normalem]{ulem}
\newtheorem{example}{Example}
\usepackage{geometry}
\usepackage{enumerate}

\newcommand{\EE}{\mathbb{E}}

\newcommand{\QQ}{\mathbb{Q}}
\newcommand{\RR}{\mathbb{R}}
\newcommand{\Ee}{\mathcal{E}}
\newcommand{\Hh}{\mathcal{H}}
\newcommand{\Nn}{\mathcal{N}}
\newcommand{\Tt}{\mathcal{T}}
\newcommand{\R}{\mathrm{R}}
\newcommand{\F}{\mathrm{F}}
\newcommand{\m}{\mathrm{m}}
\newcommand{\VIX}{\mathrm{VIX}}

\author{Blanka Horvath\footnote{Blanka Horvath gratefully acknowledges financial support from the SNSF Early Postdoc. Mobility grant~165248}\\
Department of Mathematics, King's College London \\and Imperial College London
\\
b.horvath@imperial.ac.uk
 \and Antoine Jacquier \\
 Department of Mathematics, Imperial College London\\
and Baruch College, CUNY\\
a.jacquier@imperial.ac.uk
\and Peter Tankov\footnote{Part of this research was completed while
  Peter Tankov was visiting the Department of Mathematics in Imperial College London. 
  He would like to thank the Department of Mathematics, the CFM-Imperial Institute of Quantitative Finance, 
  and the Centre National de Recherche Scientifique, France, 
  for making this visit possible. 
  The research of Peter Tankov was also supported by the chair `Financial Risks' sponsored 
  by Soci\'et\'e G\'en\'erale.}\\
ENSAE ParisTech\\
peter.tankov@polytechnique.org
}
\date{\today}
\title{Volatility options in rough volatility models}
\begin{document}
\maketitle
\begin{abstract}
We discuss the pricing and hedging of volatility options in some rough volatility models. 
First, we develop efficient Monte Carlo methods and asymptotic approximations for
computing option prices and hedge ratios in models where
log-volatility follows a Gaussian Volterra process. While providing a good fit for European
options, these models are unable to reproduce the VIX option smile
observed in the market, and are thus not suitable for VIX products. To
accommodate these, we introduce the class of modulated Volterra processes, 
and show that they successfully capture the VIX smile. 

\vspace{0.2cm}

\noindent \textbf{2010 }\textit{Mathematics Subject Classification}: 60G15, 60G22, 91G20, 91G60, 91B25\\
\noindent \textbf{Keywords: }rough volatility, VIX smile, Monte Carlo, Volterra process
\end{abstract}

\section{Introduction}
In the recent years, rough stochastic volatility models in
which the trajectories of volatility are less regular than those of
the standard Brownian motion, have gained popularity among academics and practitioners. 
As shown in~\cite{bayer2016pricing, gatheral2014volatility}, replacing standard Brownian
motion by its (rough) fractional counterpart in volatility models allows to capture and explain crucial phenomena observed both in volatility time series and in the implied volatility of option prices. 
Since then, rough volatility models have become the go-to models capable of reproducing stylised facts of financial markets and of providing a unifying theory with implications branching across financial disciplines.
A growing number of research contributions has brought about
justifications for this modelling choice, rooting from market
microstructural considerations~\cite{omar2016microstructural},
to short-term calibration of the SPX smile~\cite{ALV07, BayerRoughSkew, Fukasawa, HJL2017},
hedging~\cite{Larsson, Euch1, Fukasawatalk}
up to its potential (explored in \cite{jacquier2017vix}) to provide the sought-after parsimonious model capable of jointly handling SPX and VIX derivatives; an aim that has been a central driving factor of research in volatility modelling over the past decade \cite{Akdogan, bargourleipold, BayerDMR, BergomiSmileDynamicsII, BergomiSmileDynamicsIII, Buhler, GuyonNutz, martiniessvi, OuldAly}.

From a practical perspective a natural question arises: 
What does the mantra of rough volatility mean for a trader hedging his positions? 
Due to the non-Markovian nature of the fractional driver, 
hedging under rough volatility poses a delicate challenge making even the very definition of hedging strategies difficult.
In particular, partial differential equations can no longer be used, and simulation is the only available route so far. 
Despite the availability of efficient Monte Carlo schemes~\cite{bennedsen2017hybrid, HJMTrees, mccrickerd2017turbocharging},
pricing and model calibration in rough volatility models remain time consuming;
this heavy simulation procedure can be bypassed 
for affine rough volatility models~\cite{Larsson, euch2016characteristic, GKR18, GJRSHeston}. 

We focus here on the pricing and hedging of volatility options in rough volatility models. 
First, we show that by focusing on the forward variance instead of the instantaneous
volatility, one recovers the martingale framework and in particular
the classical martingale representation property of option prices. 
This makes it possible to compute the hedge ratios, and although the model is non-Markovian, 
in many cases options can be hedged with a finite number of liquid assets, as in the classical setting. 
Our second objective is to assess the performance of rough volatility options for the calibration of VIX option smiles. 
We confirm numerically and theoretically the observation of~\cite{bayer2016pricing} that lognormal rough volatility models are unable to calibrate VIX smiles because the VIX index is very close to
lognormal in these models. 
To accommodate the VIX smiles we therefore extend the class of lognormal models 
by adding volatility modulation through an independent stochastic factor in the Volterra integral. 
The independence of this additional factor preserves part of the analytical tractability of the lognormal setting,
and we are therefore able to develop approximate option pricing and calibration algorithms based on Fourier transform techniques. 
Using real VIX implied volatility data, we show that this new class of models is able to fully capture the skew of VIX options. 

The rest of the paper is structured as follows. 
In Section~\ref{toy.sec} we focus on a toy example where the log volatility follows a Gaussian process, 
and consider an option written on the instantaneous volatility. 
In spite of the absence of any Markovian structure, perfect hedging with a single risky asset is possible, 
and the option price is given by the Black-Scholes formula. 
Armed with this knowledge, in Section~\ref{volterra.sec}, we consider more realistic VIX index options in lognormal volatility models, and again, show that perfect hedging is possible. 
Although explicit formulas for option prices are not available in this case, 
we propose a very efficient Monte Carlo algorithm, and show that the Black-Scholes formula still
gives a good approximation to the option price. 
A drawback, however, is that lognormal volatility models are unable to capture the smile observed in the VIX option market,
so in Section~\ref{modulated.sec} we propose a new class of models to include stochastic volatility modulation, for which we develop efficient calibration strategies, and test them on real market data.

\section{A toy example: instantaneous forward variance in lognormal volatility models}\label{toy.sec}
We first consider here a simple example, namely pricing options on the instantaneous forward variance in 
lognormal (rough or not) volatility models, and show that despite the absence of Markovianity for the
volatility, pricing reduces to the Black-Scholes framework.
We assume that the instantaneous volatility process is given by        
$$
\sigma_t = \Xi e^{X_t},
$$
where $X$ is a centered Gaussian process on $\RR$ under the risk-neutral
probability, and~$\Xi$ a strictly positive constant.
For all $s\geq 0$, let $\mathcal F^0_s := \sigma(X_r,r\leq s)$, 
and $\mathcal F_s := \cap_{s<t} \mathcal F^0_t$. 
The interest rate is taken to be zero. 
Fix a time horizon~$T$, let $Z_t(T) := \EE[X_T|\mathcal F_t]$,
so that~$(Z_t(T))_{t\geq 0}$ is a Gaussian martingale and thus a process with independent increments~\cite[Theorem II.4.36]{jacod2013limit}), completely characterised by the function 
$$
c(t) := \EE[Z_t(T)^2] = \EE[\EE[X_T|\mathcal F_t]^2]. 
$$
If we assume in addition that $c(\cdot)$ is continuous then~$(Z_t(T))_{t\geq 0}$ is almost surely continuous.
Using the total variance formula, the forward variance can be characterised as
$$
\xi_t := \EE[\sigma_T^2|\mathcal F_t] = \Xi^2\EE[e^{2X_T}|\mathcal F_t]
 = \Xi^2 e^{2\EE[X_T|\mathcal F_t]+2 \text{Var}[X_T|\mathcal F_t]}
 =  \Xi^2 \exp\left\{2 (Z_t(T) + \EE[X_T^2] - c(t))\right\}.
$$
The time-$t$ price of a Call option expiring at $T_0<T$ on the instantaneous forward variance 
is given by
$P_{t}:=\EE[(\xi_{T_0}-K)^+ |\mathcal F_t]$.
Note that $(\xi_t)_{t\geq 0}$ is a continuous lognormal martingale with 
$\EE[\xi_{T_0}|\mathcal F_t] = \xi_t$
and, by the total variance formula,
$$
\text{Var}[\log \xi_{T_0}|\mathcal F_t] = 4\text{Var}[\EE[X_{T}|\mathcal F_{T_0}]|\mathcal F_t]= 4(c(T_0) - c(t)).
$$
In other words, $P_t = P(t,\xi_t)$, where $P$ is a deterministic function given by
$$
P(t,x) = \EE\left[\left(x e^{Y - \frac{1}{2}\text{Var}(Y)}-K\right)^+\right],
$$
and~$Y$ a centered Gaussian random variable with variance $4(c(T_0) - c(t))$.
Black-Scholes formula then yields
$P_t = \xi_t \Nn(d^1_t) - K \Nn(d^2_t)$,
where $\Nn$ is the standard Normal distribution function, and 
$$
d^{1,2}_t = \frac{\frac{1}{2}\log\frac{\xi_t}{K}\pm (c(T_0) - c(t))}{\sqrt{c(T_0) - c(t)}}.
$$
Applying It\^o's formula ($\xi_t$ is a continuous martingale!) 
and keeping in mind the martingale property of the option price, we obtain
$$
dP_t = \Nn(d^1_t) d\xi_t.
$$
Therefore, the forward variance option may be hedged perfectly by a portfolio containing the instantaneous variance swap and the risk-free asset. 
This happens because the forward variance process is a time-inhomogeneous geometric Brownian motion 
and therefore a Markov process in its own filtration. 
In the rest of this paper we show that perfect hedging with a finite number of assets is possible
for more complex products. 
For reasons of analytical tractability, we focus on the class of Gaussian processes
 which may be represented in the form of integrals with respect to a finite-dimensional standard Brownian motion, called Volterra processes. 
These processes allow, in particular, to easily introduce correlation between stock price and volatility.

\section{Lognormal rough Volterra stochastic volatility models}\label{volterra.sec}
The previous section was a simple framework and only considered options on the instantaneous forward variance.
We now dive deeper into the topic, and consider, still in the context of lognormal (rough) volatility models, 
options on the integrated forward variance, in particular VIX options. 
To do so, we consider volatility processes of the form
\begin{equation}\label{volterra.eq}
\sigma_t =  \Xi(t) \exp\{X_t\} 
\qquad\text{with}\qquad 
X_t = \int_{0}^t g(t,s)^\top dW_s,
\end{equation}
where $W$ is a $d$-dimensional standard Brownian motion
with respect to the filtration $\mathbb F\equiv (\mathcal F_t)_{t\in \RR}$, 
and~$g$ is a kernel satisfying the integrability condition
\begin{equation}\label{intg.eq}
\int_{0}^t \|g(t,s)\|^2 ds <\infty,\quad \text{for all }t\geq 0.
\end{equation}
The Gaussian process~$X$ in~\eqref{volterra.eq} is a Gaussian Volterra process. 
The representation given here is rather general since it is a particular case of the so-called canonical representation 
of Gaussian processes~\cite[Paragraph VI.2]{hida93gaussian}. 
Every continuous Gaussian process satisfying certain regularity assumptions admits such a representation
~\cite[Theorem 4.1]{hida93gaussian}.
Here, $\Xi(\cdot)$ is a {locally square integrable} deterministic function enabling the exact calibration of the initial variance curve.\footnote{The exact formulation from~\cite{bayer2016pricing} uses the Dol\'eans-Dade exponential instead of the simple exponential. 
The two expressions are equivalent, as the additional deterministic term can be included in~$\Xi$.}
The Mandelbrot-van Ness formulation~\cite{Mandelbrot} of fractional Brownian motion requires the integral to start from~$-\infty$ instead of~$0$.
Since the volatility process is taken conditional on the pricing time zero, 
the two formulations are in fact equivalent, and~$\Xi$ takes into account the past history 
of the process.
The formulation~\eqref{volterra.eq} extends the so-called rough Bergomi model introduced in~\cite{bayer2016pricing}.
The rough Bergomi model
corresponding to a one-dimensional Brownian motion~$W$ and a function~$g$ of the form
\begin{equation}\label{eq:g}
g(t,s) = \alpha(t-s)^{H-\frac{1}{2}},
\qquad\text{for }s\in [0,t),
\qquad\text{with}\qquad
\alpha = 2\nu\sqrt{\frac{\Gamma(3/2-H)}{\Gamma(H+1)\Gamma(2-2H)}},
\end{equation}
where $\nu>0$ is the volatility of volatility 
and $H \in (0,1)$ the Hurst parameter of the fractional Brownian motion. 
This kernel~$g(\cdot)$ clearly satisfies the integrability condition~\eqref{intg.eq}.

Our goal here is to develop the theory and provide numerical
algorithms for pricing and hedging options in generic models of the form~\eqref{volterra.eq}. 
Empirical analysis of forward volatility curves
with the aim of choosing the adequate number of factors $d$ and the
suitable shapes of the kernel function $g$ is the topic of our ongoing
research. 
Similarly to the toy example of the previous section, we introduce the
martingale framework by considering the conditional expectation
process, which is given, for any $t\leq u$, by
$$
Z_t(u): = \EE[X_u|\mathcal F_t] =  \int_{0}^t g(u,s)^\top dW_s.
$$
Therefore the forward variance $\xi_t(u):=\EE[\sigma_u^2|\mathcal F_t]$ has the explicit martingale dynamics
\begin{equation}\label{eq:DynamicsXi}
d\xi_t(u) = 2\xi_t(u)  g(u,t)^\top dW_t,
\qquad\text{for } t\leq u. 
\end{equation}

We are interested here in pricing an option with pay-off at time $T$ given by 
\begin{equation}\label{eq:fPayoff}
f\left(\frac{1}{\Theta}\int_{T}^{T+\Theta} \xi_{T}(u) du\right),
\end{equation}
for some time horizon $\Theta$.
Since the value of the VIX index at time~$T$ can be computed via the continuous-time monitoring formula
\begin{equation}\label{eq:VIXFormula}
\VIX_{T} := \sqrt{\frac{1}{\Theta}\int_{T}^{T+\Theta} \xi_{T}(u) du},
\end{equation}
with $\Theta$ being one month, the Call option on the VIX corresponds to 
$f(x) = (\sqrt{x}-K)^+$. 
The time-$t$ price of such an option is given by 
$$
P_t := \EE\left[f\left(\frac{1}{\Theta}\int_{T}^{T+\Theta} \xi_{T}(u) du\right) \Big|\mathcal
  F_t\right]  = F(t,\xi_t(u)_{T\leq u \leq T+\Theta}),
$$
where $F$ is a deterministic mapping from $[0,T]\times \Hh$,
with $\Hh := L^2([T,T+\Theta])$ to $\RR$, defined by 
\begin{equation}\label{vixasian}
F(t,x) := \EE\left[f\left(\frac{1}{\Theta}\int_{T}^{T+\Theta}x(u) \Ee_{t,{T}}(u) du\right) \right],
\end{equation}
where
\begin{equation}\label{eq:DDExp}
\Ee_{t,T}(u):= \Ee\left(2\int_t^\cdot g(u,s)^\top
  dW_s\right)_{T}   = \exp\left(2\int_t^{T} g(u,s)^\top
  dW_s - 2\int_t^{T} \|g(u,s)\|^2 ds\right)
\end{equation}
is the Dol\'eans-Dade exponential. This representation allows to
easily derive the hedging strategy for such a product, as discussed in
the following section.

\subsection{Martingale representation and hedging of VIX options}
The following theorem provides a martingale representation for the VIX
options, which serves as a basis for the hedging strategy. 

\begin{proposition}
Let the function $f$ be {piecewise} differentiable with $f'$ piecewise continuous
and bounded. Then the option price $P_t$ admits the martingale representation
$$
P_T = P_t {-} 2\int_t^T \int_{T}^{T+\Theta} D_x F(s,\xi_s) (u) \xi_s(u)
g(u,s)^\top du\, dW_s,
$$
where the Fr\'echet derivative $D_xF$ is given by 
$$
D_x F(t,x)(v) = \EE\left[f'\left(\frac{1}{\Theta}\int_T^{T+\Theta}x(u) \Ee_{t,T}(u) du\right) \frac{\Ee_{t,T}(v)}{\Theta} \right] 
$$
\end{proposition}
\begin{proof}
\textit{Step 1.}
Let $f_\varepsilon\in C^2(\RR)$ be a mollified version of the
function~$f$ such that~$f'_\varepsilon$ and~$f^{\prime\prime}_\varepsilon$ are bounded and
continuous, $|f_\varepsilon(x)| \leq |f(x)| + C$ for all~$x$,
$f'_\varepsilon$ { is bounded uniformly on
  $\varepsilon$ and} converges to~$f'$ at the points of continuity of $f'$ and
$f_\varepsilon(x)$ converges to~$f(x)$ for all~$x$ as $\varepsilon$ tends to zero. 
One can for example take 
$$
f_\varepsilon(x) = \int_{\RR} f(x+z) p_\varepsilon (z) dz, 
$$
where $p_\varepsilon$ is a family of smooth compactly supported
densities converging to the delta function as $\varepsilon$ approaches zero. 
The first step is to prove the martingale representation for the price of the option
with pay-off~$f_\varepsilon$ by applying the infinite-dimensional It\^o formula.
Let $(v,w) \in [T,T+\Theta]^2$, and define
$$
F_\varepsilon(t,x) := \EE\left[f_\varepsilon\left(\frac{1}{\Theta}\int_T^{T+\Theta}x(u) \Ee_{t,T}(u)du\right) \right],
$$ 
with $\Ee_{t,T}(\cdot)$ defined in~\eqref{eq:DDExp}.
For $h\in \Hh$, {the mean value theorem implies the existence of $\theta_{\delta} \in [0,1]$ such that}
\begin{align*}
\lim_{\delta\downarrow 0} \frac{F_\varepsilon(t,x+\delta h) -
  F_\varepsilon(t,x)}{\delta}
   &= \lim_{\delta \downarrow 0} \EE\left[f'_\varepsilon\left(\frac{1}{\Theta}\int_T^{T+\Theta}(x(u) + \theta_{\delta} \delta h(u))\Ee_{t,T}(u) du\right) \frac{1}{\Theta}\int_T^{T+\Theta}\Ee_{t,T}(v) h(v) dv \right] \\
& =  \frac{1}{\Theta}\int_T^{T+\Theta} h(v) dv\EE\left[f'_\varepsilon\left(\frac{1}{\Theta}\int_T^{T+\Theta}x(u)\Ee_{t,T}(u) du\right)\Ee_{t,T}(v) \right],
\end{align*}
where we have used the dominated convergence theorem and Fubini's Theorem. 
Since the expectation under the integral is bounded, 
the Fr\'echet derivative of $F_\varepsilon$ is then given by 
\begin{align*}
D_x F_\varepsilon(t,x)(v)
 = \EE\left[f'_\varepsilon\left(\frac{1}{\Theta}\int_T^{T+\Theta}x(u)\Ee_{t,T}(u) du\right)
 \frac{\Ee_{t,T}(v)}{\Theta}  \right]
  \in \Hh.
\end{align*}
Moreover, a similar argument shows that it is uniformly continuous in $x$ on $[0,T]\times \Hh$. 
Iterating the procedure, we find the second Fr\'echet derivative
\begin{align*}
D^2_{xx} F_\varepsilon(t,x)(v,w)
 = \EE\left[f_\varepsilon^{\prime\prime}\left(\frac{1}{\Theta}\int_T^{T+\Theta}x(u)\Ee_{t,T}(u)du\right)
 \frac{\Ee_{t,T}(v) \Ee_{t,T}(w)}{\Theta^2} \right] \in \mathcal{L}(\Hh,\Hh),
\end{align*}
which is also uniformly continuous, where $\mathcal{L}(\Hh, \Hh)$ denotes the class of linear operators from~$\Hh$ to~$\Hh$.
Finally, for the derivative of~$F_\varepsilon$ with respect to~$t$, we can write
$$
\frac{F_\varepsilon(t+\delta,x) - F_{\varepsilon}(t,x)}{\delta} =
{-} \frac{1}{\delta}\EE\left[f_\varepsilon\left(\frac{1}{\Theta}\int_T^{T+\Theta} x(u)X^u Y^u_{t+\delta} du\right)
 - f_\varepsilon\left(\frac{1}{\Theta}\int_T^{T+\Theta} x(u)X^u Y^u_{t} du\right)\right],
$$
with 
$X^u := \Ee_{t+\delta,T}(u)$
and
$Y^u_r := \Ee_{t,r}(u)$.
For $r\leq t+\delta$, $X^u$ and $Y^u_r$ depend on the increments of~$W$ 
over disjoint intervals, and so are independent;
hence the infinite-dimensional It\^o formula~\cite[Part~I, Theorem~4.32]{daprato2014stochastic} with respect to~$Y$, keeping~$X$ constant, yields
\begin{align*}
&f_\varepsilon\left(\frac{1}{\Theta}\int_T^{T+\Theta} x(u)X^u Y^u_{t+\delta} du\right)
 - f_\varepsilon\left(\frac{1}{\Theta}\int_T^{T+\Theta} x(u)X^u Y^u_{t} du\right)\\
& = 2\int_t^{t+\delta} f'_\varepsilon\left(\frac{1}{\Theta}\int_T^{T+\Theta} x(u)X^u Y^u_{r} du\right) 
\frac{1}{\Theta}\int_T^{T+\Theta} x(v) X^v Y^v_r g(v,r)^\top  dv dW_r\\
& + \frac{2}{\Theta^2} \int_t^{t+\delta} f^{\prime\prime}_\varepsilon \left(\frac{1}{\Theta}\int_T^{T+\Theta} x(u)
    X^u Y^u_{r} du\right) \int_T^{T+\Theta} \int_T^{T+\Theta} x(v) x(w)
  X^v Y^v_r X^w Y^w_r g(v,r)^\top g(w,r)  dv  dw dr.
\end{align*}
After taking the expectation, the first term disappears due to the independence of~$X$ and~$Y$. 
Dividing by~$\delta$ and taking the limit as~$\delta$ tends to zero, we then obtain, by dominated convergence,
$$
D_t F_\varepsilon (t,x)
 = {-} 2\int_T^{T+\Theta}\int_T^{T+\Theta} 
 x(v) x(w) g(v,t)^\top g(w,t) D^2_{xx} F_\varepsilon(t,x) (v,w) dv dw. 
$$
{It follows that 
$$
D_t F_\varepsilon (t,\xi_t) dt =
-\frac{1}{2}\langle d\xi_t, D^2_{xx} F_\varepsilon(t,x) d\xi_t\rangle,
$$
as expected from the local martingale property of $F_\varepsilon(t,\xi_t)$. }
Now, applying It\^o's formula, we obtain a martingale
representation for the regularised option price
$$ 
F_\varepsilon(T,\xi_T) = F_\varepsilon(t,\xi_t)
 {-} 2\int_t^T\int_{T}^{T+\Theta} D_x F_\varepsilon(s,\xi_s) (u) \xi_s(u)g(u,s)^\top du\, dW_s.  
$$  
\noindent \textit{Step 2. } 
It remains to pass to the limit as $\varepsilon$ tends to zero. 
By dominated convergence, $F_\varepsilon(T,\xi_T)$ converges to $F(T,\xi_T) $ 
and $F_\varepsilon(t,\xi_t)$ to $F(t,\xi_t) $. 
For the convergence of the
stochastic integral, we apply also dominated convergence. 
First, observe that for $x\neq0$, the random variable
$\int_T^{T+\Theta}x(u)\Ee_{t,T}(u)$
has no atom (Lemma~\ref{density.lm} in the appendix).
Then, dominated convergence allows us to conclude that for every $t,x$ and $v$, 
$\lim_{\varepsilon\downarrow 0}D_x F_{\varepsilon}(t,x)(v) = D_xF(t,x)(v)$. 
In addition, this derivative is bounded {by $\frac{\sup_{\varepsilon,x}
|f'_\varepsilon(x)|}{\Theta}$,} which means that the inner integral 
$$
\int_{T}^{T+\Theta} D_x F_\varepsilon(s,\xi_s) (u) \xi_s(u) 
g(u,s)^\top du
$$
converges {and admits the following bound: 
$$
\left|\int_{T}^{T+\Theta} D_x F_\varepsilon(s,\xi_s) (u) \xi_s(u) 
g(u,s)^\top du\right| \leq \int_T^{T+\Theta} \frac{\sup_{\varepsilon,x}
|f'_\varepsilon(x)|}{\Theta} \xi_s(u) \|g(u,s)\|ds.
$$
This in turn implies that we can use 
dominated convergence for stochastic integrals (Theorem IV.32 in \cite{protter}) to prove the convergence of the outer integral. }
\end{proof}

\subsection{Pricing VIX options by Monte Carlo}\label{montecarlo.sec}
\subsubsection{Discretisation schemes}
From~\eqref{vixasian}, the computation of a forward variance option price in a
lognormal stochastic volatility model reduces to the computation of
the expected functional of an integral over a family of correlated
lognormal random variables. 
To compute this option price by Monte Carlo, we approximate~\eqref{vixasian} 
by discretising the integral. 
We shall consider the discretisation grid
\begin{equation}\label{eq:Grids}
\Tt_{\kappa}:= \left\{t_i^n = T + \Theta \left(\frac{i}{n}\right)^\kappa\right\}_{i=0,\ldots,n}
\qquad\text{for }\kappa>0,
\end{equation}
and the following two discretisation schemes: 
\begin{itemize}
\item the rectangle scheme on~$\Tt_{1}$
(with $\zeta^n_i := \int_{t^n_i}^{t^n_{i+1}} x(u) du$ and
$\eta^{n}(u) := \max\{t^n_i: t^n_i\leq u\}$):
$$
F_n(t,x) :=\EE\left[f\left(\frac{1}{\Theta}\sum_{i=0}^{n-1}\zeta^n_i  \Ee_{t,T}(t^n_i)\right) \right] = \EE\left[f\left(\frac{1}{\Theta}\int_{T}^{T+\Theta}x(u) \Ee_{t,T}(\eta^n(u)) du\right) \right];
$$
\item the trapezoidal scheme on~$\Tt_{\kappa}$
(with $\theta^n(u) := \frac{t^n_{i+1}-u}{t^{n}_{i+1} - t^n_i}$ and
$\overline{\eta}^n(u) := \min\{t^n_i: t^n_i> u\}$):
\begin{align*}
\widehat{F}_n(t,x)
 & := \EE\left[f\left(\frac{1}{\Theta}\sum_{i=0}^{n-1}\int_{t^n_i}^{t^n_{i+1}} x(u) (
    \theta^n (u)\Ee_{t,T}(t^n_i) + (1-\theta^n(u)) \mathcal
  E_{t,T}(t^n_{i+1}) )du\right) \right]\\
& = \EE\left[f\left(\frac{1}{\Theta}\int_{T}^{T+\Theta} x(u) (
    \theta^n (u)\Ee_{t,T}(\eta^n(u)) + (1-\theta^n(u)) \mathcal
  E_{t,T}(\overline{\eta}^n(u)) )du\right) \right].
\end{align*}
\end{itemize}
Note that the sequence of random variables $(Z_i)_{i=0}^{n-1}$ defined by
\begin{equation}\label{eq:ZProcess}
Z_{i}:= \log \Ee_{t,T}(t^n_i) = 2\int_t^T g(t^n_i,s)^\top dW_s - 2\int_t^T \|g(t^n_i,s)\|^2 ds
\end{equation}
forms a Gaussian random vector with mean
$m_i = - 2\int_t^T \|g(t^n_i,s)\|^2 ds$
and covariance matrix 
$$
C_{ij} := \mathrm{Cov}[Z_{i},Z_{j}] = 4 \int_t^T  g(t^n_i,s)^\top g(t^n_j,s) ds.
$$
The following proposition, proved in Appendix~\ref{app:Proofmc.prop}, characterises the convergence rate of these~schemes. 
\begin{proposition}\label{mc.prop}
Let $f$ be Lipschitz, $x(\cdot)$ bounded, and assume there exist
$\beta,c>0$ such that
$$
\left(\int_t^T \left\|g(t_2,s)-g(t_1,s)\right\|^2ds\right)^{1/2}\leq c(t_2-t_1) (t_2-T)^{\beta-1},
\quad\text{for all }T\leq t_1 <t_2.
$$
\begin{itemize}
\item For the rectangle scheme on~$\Tt_{1}$, 
$|F(t,x)-F_n(t,x)|= \mathcal{O}\left(\frac{1}{n}\right)$;
\item if in addition, for all $T\leq t_1 \leq t_2<t_3$,
$$
\left(\int_t^T \left\|g(t_2,s)-\frac{t_3-t_2}{t_3-t_1}g(t_1,s)- \frac{t_2-t_1}{t_3-t_1}g(t_3,s)\right\|^2ds\right)^{1/2}\leq c(t_3-t_1)^2 (t_3-T)^{\beta-2},
$$ 
then for the trapezoidal scheme on~$\Tt_{\kappa}$ with $\kappa(\beta+1)>2$, 
$|F(t,x)-\widehat F_n(t,x)|= \mathcal{O}\left(\frac{1}{n^2}\right)$. 
\end{itemize}
\end{proposition}

\subsubsection{Control variate}\label{sec:ControlVarLN}
In our model, the squared VIX index 
$\VIX^2_T := \frac{1}{\Theta} \int_{T}^{T+\Theta} \xi_T (u) du$
is an integral over a family of lognormal random variables. 
Mimicking Kemna and Vorst~\cite{Kemna}'s control variate trick 
(originally proposed in the context of Asian options in lognormal models), 
it is natural to approximate this integral by the exponential of an
integral of a corresponding family of Gaussian random variables:
\begin{equation}\label{eq:ApproxVIXSquared}
\overline{\VIX}^2_T := \exp\left(\frac{ 1}{\Theta}\int_{T}^{T+\Theta} \log  \xi_T (u) du \right).
\end{equation}
The corresponding approximation for the VIX option price,
which can be used as a control variate in the Monte Carlo scheme, is given by 
$$
{\overline{P}_t
 := \EE\left[f\left(\overline{\VIX}^2_T\right) \Big|\mathcal F_t\right]
 =: \overline{F}\left(t,\xi_t(u)_{T\leq u \leq T+\Theta}\right),}
$$
with  
$\overline{F}(t,x) := \EE\left[f\left(e^{Y}\right)\right]$,
and~$Y$ Gaussian with mean~$m_{Y}$ and variance~$\sigma^2_{Y}$
given by
\begin{align*}
m_{Y} &= \frac{1}{\Theta} \int_{T}^{T+\Theta} \left\{\log x(u)-2 \int_t^{T}
\|g(u,s)\|^2 ds\right\}du, 
\\
\sigma_{Y}^2 &= \frac{4}{\Theta^2}\int_{[T,{T+\Theta}]^2}du dv \int_t^{T} g(u,s)^\top g(v,s) ds. 
\end{align*}
For a Call option on VIX, $f(x) = (\sqrt{x}-K)^+$,  the approximate price is therefore given by
\begin{align*}
{\overline{P}_t} = \EE\left[\left(e^{\frac{1}{2}Y}-K\right)^+\right]
 = \overline{Y} \Nn\left(\frac{\log\frac{\overline{Y}}{K} + \frac{1}{8}\sigma_{Y}^2}{\sigma_{Y}/2}\right)
  - K \Nn\left(\frac{\log\frac{\overline{Y}}{K} - \frac{1}{8}\sigma_{Y}^2}{\sigma_{Y}/2}\right),
\end{align*}
with $\overline{Y} = \exp\left(\frac{1}{2}m_{Y} +  \frac{\sigma_{Y}^2}{8}\right)$.

\subsection{Numerical illustration}\label{bergomi.sec}
Consider the model introduced in~\eqref{volterra.eq}, together with the characterisation~\eqref{intg.eq},
essentially the rough Bergomi model from~\cite{bayer2016pricing}.
From~\eqref{eq:DynamicsXi}, the dynamics of the forward variance is
$$
\frac{d\xi_t(T)}{\xi_t(T)} = 2\alpha(T-t)^{H-\frac{1}{2}} dW_t,
  \quad \text{for all }0\leq t\leq T.
$$
Since all the forward variances are driven by the same Brownian
motion, it is enough to invest into a single variance swap to achieve perfect hedging. 
As an example, consider an option with pay-off
$$
f\left(\frac{1}{\Theta}\int_{T}^{{T+\Theta}} \xi_{T}(u) du\right),
$$
and denote by $F_t$ its price at time $t$, with $F_t = F(t,\xi_t)$,
with~$F$ as in~\eqref{vixasian}. 
Using the Monte Carlo method described in Section~\ref{montecarlo.sec}, the option price is approximated by 
$$
F_n(t,x) = \EE\left[f\left(\frac{1}{\Theta}\sum_{i=0}^{n-1} \zeta^n_i e^{Z_i}\right)\right]
$$
in the rectangle discretisation or by 
$$
\widehat F_n(t,x) = \EE\left[f\left(\frac{1}{\Theta}\sum_{i=0}^{n-1}
    \int_{t^n_i}^{t^n_{i+1}} x(u) (\theta^n(u) e^{Z_i} + (1-\theta^n(u))e^{Z_{i+1}})\right)\right]
$$
in the trapezoidal discretisation, where $(Z_i)_{i=0}^{n-1}$ is a Gaussian random vector with mean
$$
m_i = -2 \alpha^2\int_t^{T} |t_i-s|^{2H-1} ds = -\alpha^2\frac{|t_i-T|^{2H} - |t_i-t|^{2H}}{H}
$$
and covariance matrix $(C_{ij})_{i,j=0,\ldots,n-1}$, which can be computed as follows:
when $i<j$, 
\begin{align*}
C_{ij} &= 4\alpha^2 \int_t^{T} (t_i - s)^{H-\frac{1}{2}} (t_j -
s)^{H-\frac{1}{2}} ds \\ &= 4\alpha^2 \int_0^{t_i-t}(t_j-t_i +
s)^{H-\frac{1}{2}} s^{H-\frac{1}{2}} ds - 4\alpha^2 \int_0^{t_i-{T}}(t_j-t_i +
s)^{H-\frac{1}{2}} s^{H-\frac{1}{2}} ds\\
& = \frac{4\alpha^2}{H+\frac{1}{2}} (t_j-t_i)^{H-\frac{1}{2}}
(t_i-t)^{H+\frac{1}{2}}
{}_{2}F_1\left(\frac{1}{2}-H,\frac{1}{2}+H,\frac{3}{2}+H,-\frac{t_i-t}{t_j-t_i}\right)\\
&-\frac{4\alpha^2}{H+\frac{1}{2}} (t_j-t_i)^{H-\frac{1}{2}}
(t_i-{T})^{H+\frac{1}{2}} {}_{2}F_1\left(\frac{1}{2}-H,\frac{1}{2}+H,\frac{3}{2}+H,-\frac{t_i-{T}}{t_j-t_i}\right),
\end{align*}
where ${}_2F_1$ is the Hypergeometric function~\cite[integral 3.197.8]{grad}, and when $i=j$, 
$$
C_{ii} = 4\alpha^2 \int_t^{T} (t_i - s)^{2H-1} ds = \frac{2\alpha^2}{H}\left[(t_i-t)^{2H}-(t_i-{T})^{2H}\right].
$$
The following corollary of Proposition~\ref{mc.prop} provides the
convergence rates of the Monte Carlo scheme. 
To simplify notation, we assume that $t\geq 0$ for the rest of this section. 
\begin{corollary}\label{mcrb1.corr}
Let $f$ be Lipschitz and assume that $x(\cdot)$ is bounded.
In the rough Bergomi model,
\begin{itemize}
\item for the rectangle scheme on~$\Tt_{1}$, we have
$|F(t,x) - F_n(t,x)|= \mathcal{O}\left(\frac{1}{n}\right)$;
\item for the trapezoidal scheme on~$\Tt_{\kappa}$ (with
  $\kappa(H+1)>2$), we have
$\left|F(t,x) - \widehat F_n(t,x)\right|= \mathcal{O}\left(\frac{1}{n^2}\right)$.
\end{itemize}
\end{corollary}

Figures~\ref{rbergomi.fig} and~\ref{rbergomi2.fig} illustrate the convergence of the Monte Carlo estimator for the VIX option price. 
The model parameters are $\alpha=0.2$, $H=0.1$, 
flat forward variance with $\sqrt{\xi_0}=20\%$, 
time to maturity $T = 1$ year, and $\Theta =0.1$. 
We used $50,000$ paths for all calculations.  {To decrease variance, we
use the discretised version of the control variate from
Section~\ref{sec:ControlVarLN}, meaning that the integral in
\eqref{eq:ApproxVIXSquared} is replaced with the sum obtained using
the corresponding discretisation scheme. This replacement is done both
in the Monte Carlo estimator and in the exact computation, so that the
control variate introduces no additional bias. To choose the
discretisation dates for the trapezoidal scheme, we took $\kappa=2$. }
Figure~\ref{rbergomi_smile.fig} plots the implied volatility smiles in
the rough Bergomi model for different parameter values 
(parameters not mentioned in the graphs are the same as above). 
The implied volatility of VIX options is defined assuming that the VIX future is lognormal and
using the model-generated VIX future as initial value. 
In this model where the volatility is lognormal, the VIX smile is almost completely flat.
This is of course due to the fact that the VIX itself is almost lognormal in the model, 
because the averaging interval~$\Theta$ (one month) 
over which the VIX is defined is rather short
(in fact, this feature can be used efficiently for approximation purposes, as in~\cite{jacquier2017vix}). 
The actual implied volatility smiles observed in the VIX option market (Figure~\ref{vix_smiles.fig}) 
are of course not flat and exhibit a pronounced positive skew. 
This inconsistency, already observed in~\cite{bayer2016pricing}, 
shows that while the rough Bergomi model fits accurately index option smiles, 
it is clearly not sufficient to calibrate the VIX smile. 
In the following section we therefore propose an extension of this model which makes such calibration possible. 
\begin{figure}
\centerline{\includegraphics[scale=0.35]{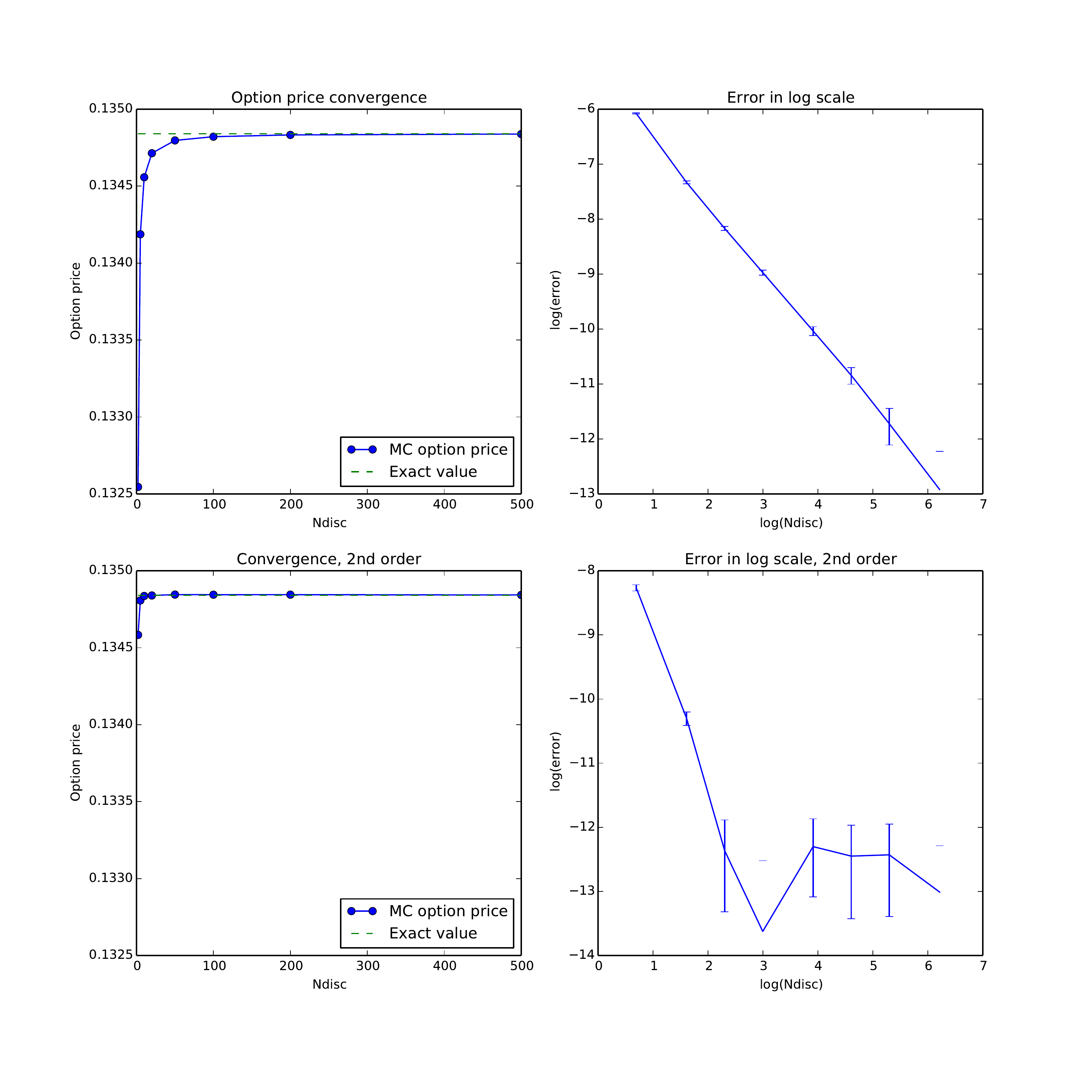}}
\caption{Convergence of the Monte Carlo estimator for the price
  of a Call option on the VIX with zero strike. 
  Top left: rectangle scheme. Bottom left: trapezoidal scheme. 
  The right graphs show the respective errors on a logarithmic
  scale. {When the confidence interval contains zero, only the upper
  bound is shown.}}
\label{rbergomi.fig}
\end{figure}

\begin{figure}
\centerline{\includegraphics[scale=0.35]{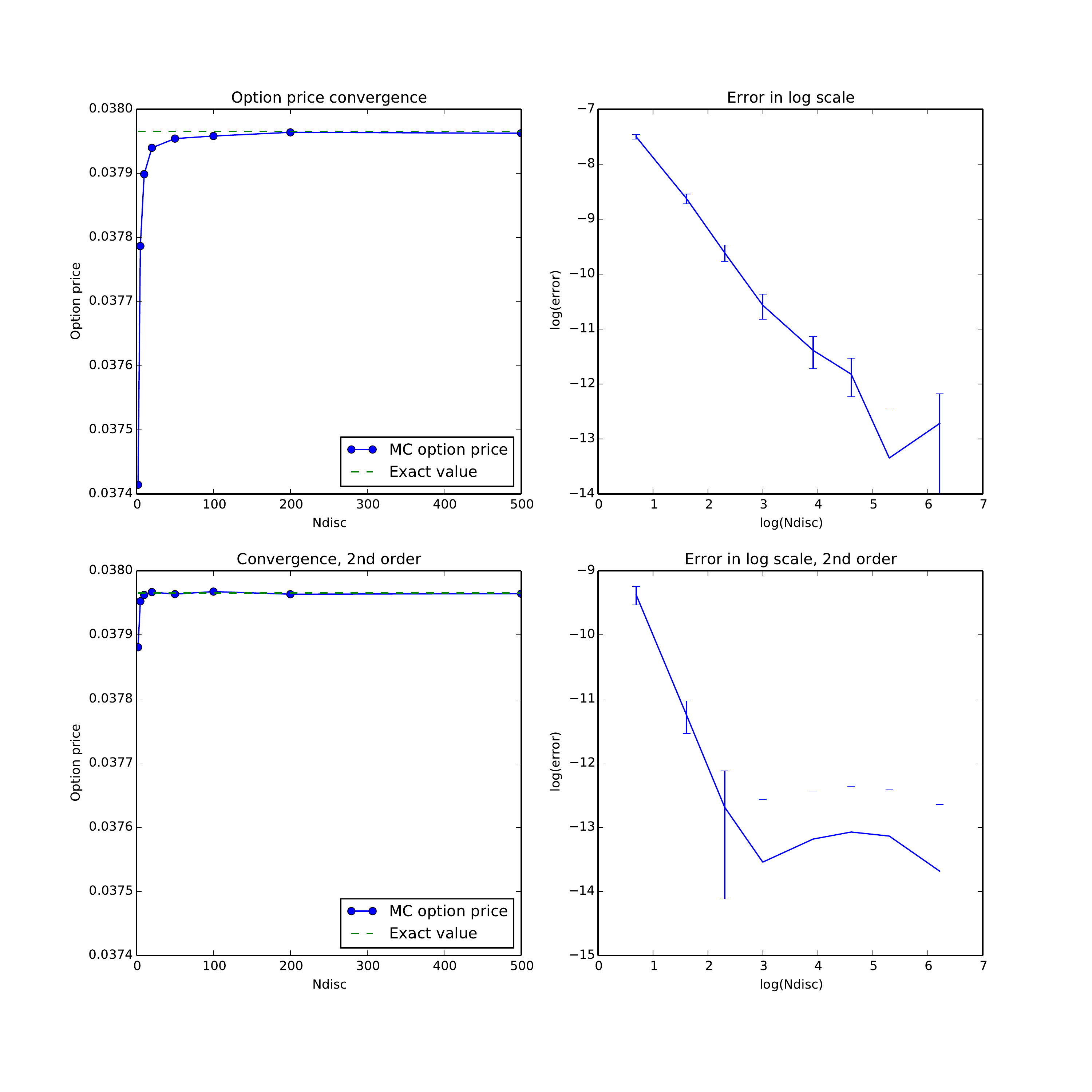}}
\caption{Convergence of the Monte Carlo estimator for the price
  of a Call option on the VIX with strike $K=0.1$. 
  Top left: rectangle   scheme. Bottom left: trapezoidal scheme. 
  The right graphs show the respective errors in the logarithmic
  scale.  {When the confidence interval contains zero, only the upper
  bound is shown.}}
\label{rbergomi2.fig}
\end{figure}

\begin{figure}
\centerline{\includegraphics[width=0.5\textwidth]{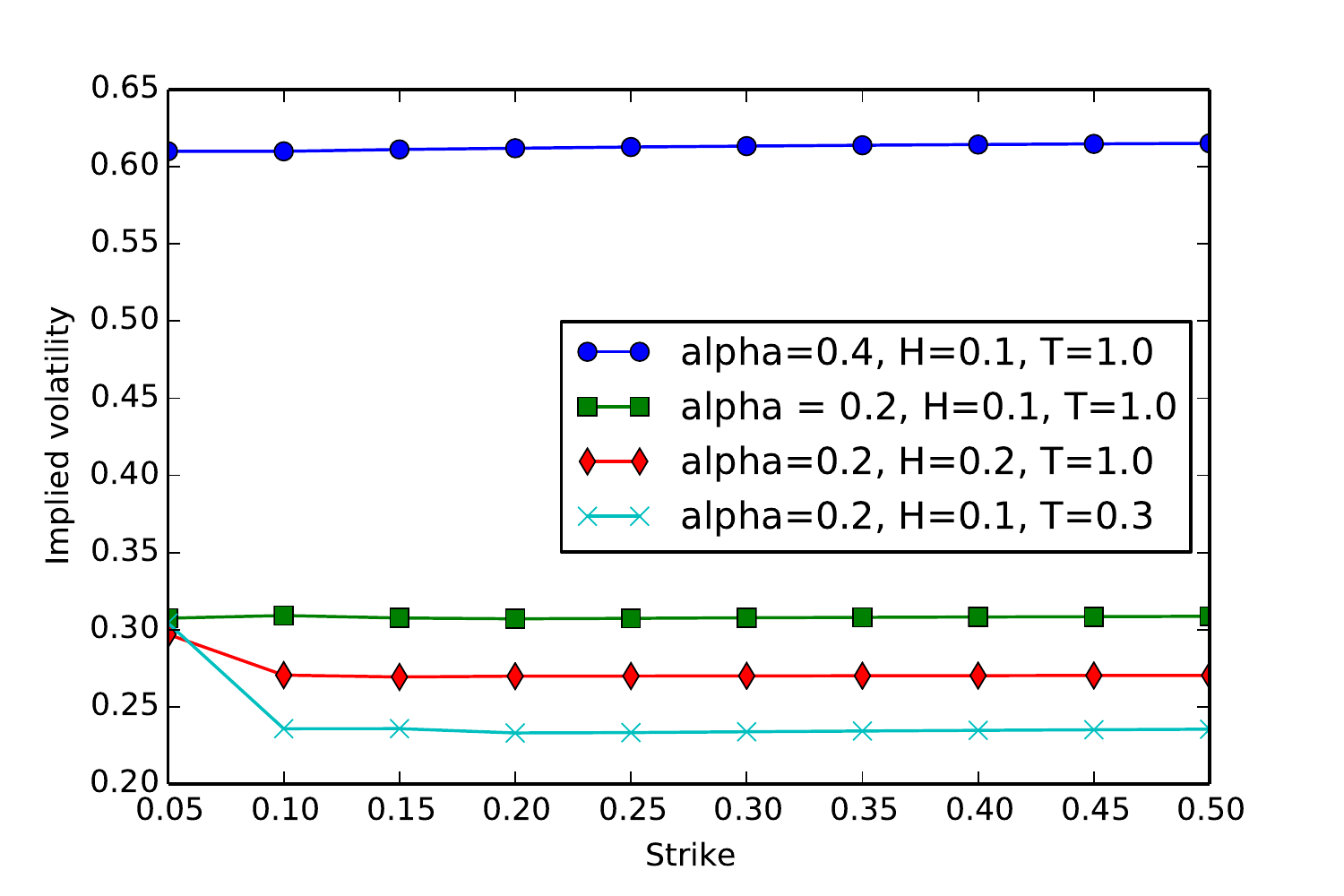}}
\caption{VIX implied smiles in rough Bergomi under different parameter combinations.}
\label{rbergomi_smile.fig}
\end{figure}

\begin{figure}
\centerline{\includegraphics[width=0.33\textwidth]{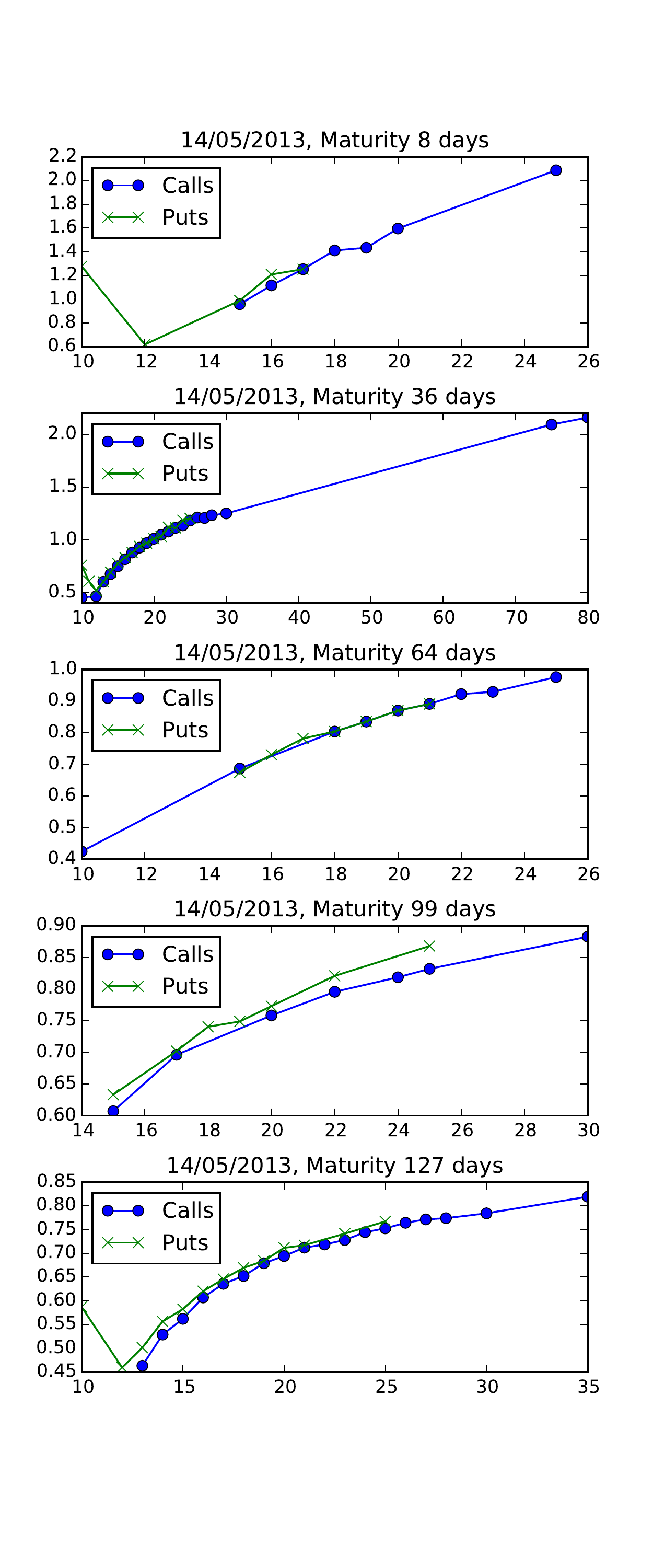}\includegraphics[width=0.33\textwidth]{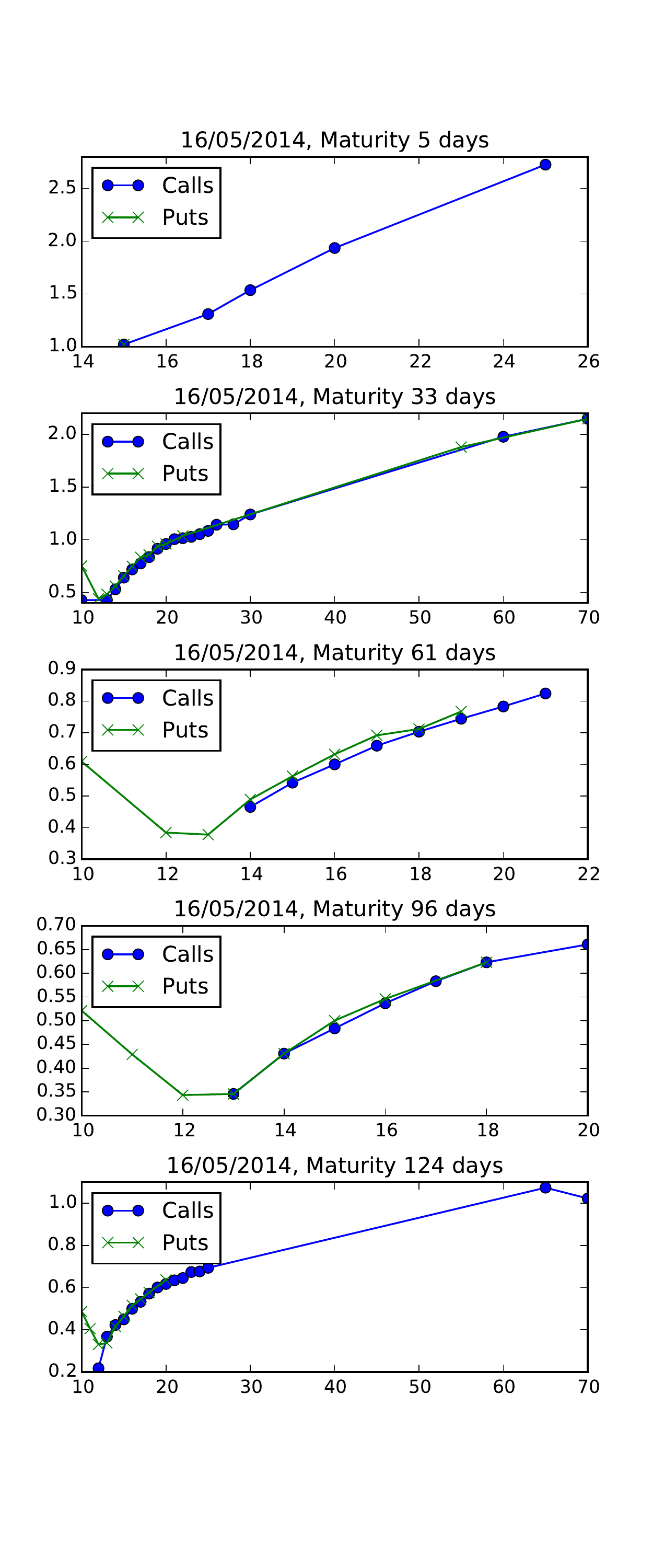}\includegraphics[width=0.33\textwidth]{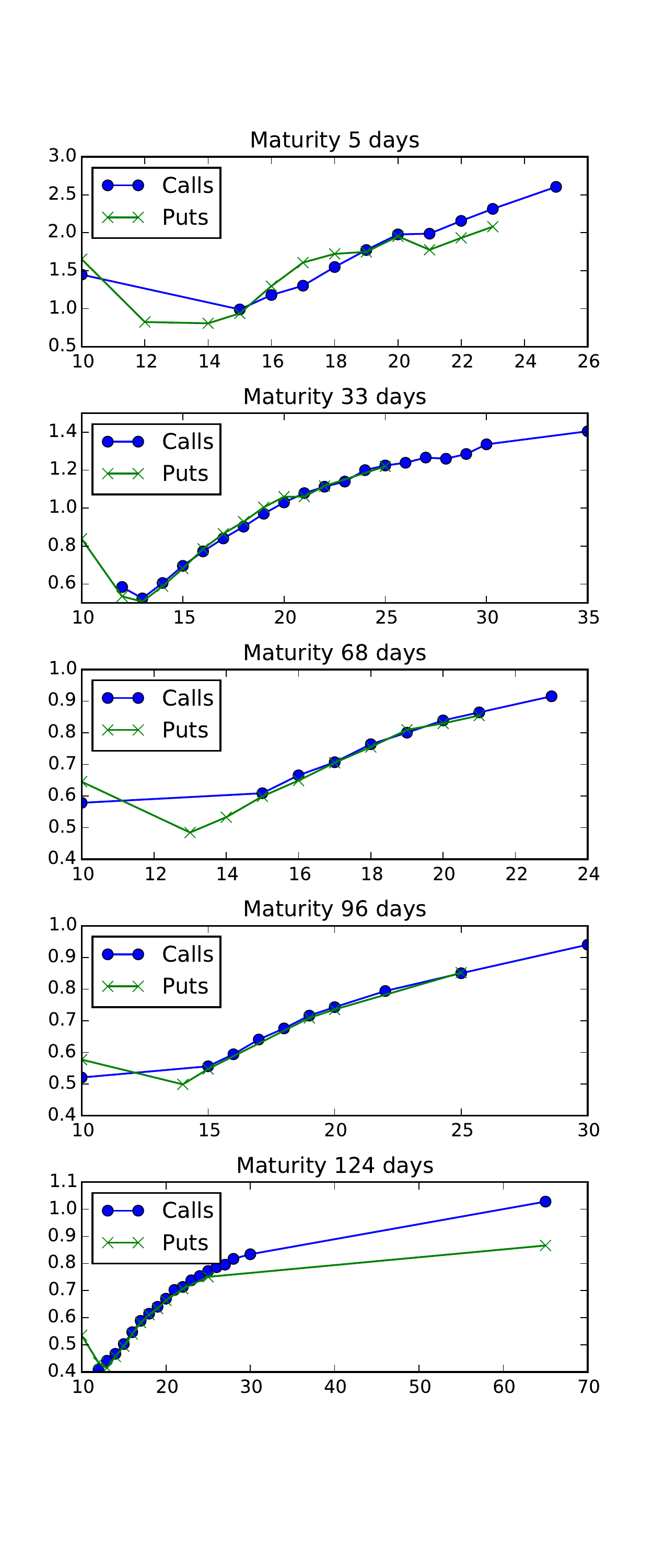}}
\caption{Actual VIX implied volatility smiles at different dates and for different maturities.}
\label{vix_smiles.fig}
\end{figure}

\newpage
\section{Modulated Volterra stochastic volatility processes}\label{modulated.sec}
We now propose a new class of rough volatility models, 
able to capture the specificities of the VIX implied volatility smile.
Specifically, we assume that the instantaneous volatility process satisfies
\begin{equation}\label{modulated.eq}
\sigma_t =   \Xi(t) e^{X_t}, 
\qquad \text{with}\qquad 
X_t = \int_{0}^t\sqrt{\Gamma_s} g(t,s)^\top dW_s,
\end{equation}
where again $\Xi(\cdot)$ is a deterministic function used for the calibration of the initial forward variance curve, 
$W$ a $d$-dimensional standard Brownian motion with respect to the filtration $\mathbb F\equiv (\mathcal F_t)_{t\in \RR_+}$, 
$g$ a kernel such that 
$\int_{0}^t \|g(t,s)\|^2 ds$ is finite for all $t\geq 0$, 
and $\Gamma$ a time-homogeneous positive conservative affine process independent of~$W$. 
This model is reminiscent of Brownian semi-stationary processes, 
introduced by Barndorff-Nielsen and Schmiegel~\cite{BSS}.
However, the stochastic integral starts here at time zero, instead of~$-\infty$,
in order to avoid working with affine processes on the whole real line.
Following~\cite[Theorem 2.7 and Proposition 9.1]{duffie2003affine}, the infinitesimal generator of~$\Gamma$ takes the form
$$
\mathcal{L}f(x) = k(\theta-x)\frac{\partial f}{\partial x}(x)
 + \frac{\delta^2 x}{2} \frac{\partial^2 f}{\partial x^2}(x)
 + \int_0^\infty\{f(x+z)-f(x)\}\{\m(dz) + x\mu(dz)\},
$$
where $k, \theta, \delta\geq 0$ and $\m, \mu$ are positive measures on $(0,\infty)$ 
such that $\int_0^\infty(z\wedge 1) \{\mu(dz) + \m(dz)\}$ is finite. 
For later use, we define the functions
$$
\R(u) := -ku + \frac{\delta^2}{2} u^2 + \int_0^\infty (e^{zu}-1)\mu(dz)
\qquad\text{and}\qquad
\F(u) := k\theta u + \int_0^\infty (e^{zu}-1)\m(dz). 
$$
We further assume that the kernel~$g$ satisfies $g(t,s) = g(t-s)$ for $0\leq s\leq t$, 
and we let $G(t) := \int_0^t \|g(s)\|^2 ds$.
We next introduce the following technical assumption:
\begin{assumption}\label{assu:Levy}
For a fixed finite time horizon $T$, there exists $A>0$ such that 
$$
\int_1^\infty ze^{z A}\{\mu(dz) + \m(dz)\}<\infty
\qquad\text{and}\qquad
2G(T) + T(0\vee \R(A)) \leq A. 
$$
\end{assumption}
This framework yields the following result:
\begin{proposition}\label{affine.prop}
Under Assumption~\ref{assu:Levy}, the ordinary differential equations
$$
\partial_t \psi(t) = 2\|g(t)\|^2 + \R(\psi(t))
\qquad\text{and}\qquad
\partial_t \phi(t) = \F(\psi(t)),
$$
together with initial conditions $\psi(0) = \phi(0) = 0$, have solutions on $[0,T]$, 
with $0\leq \psi(\cdot)\leq A$ and $0\leq \phi(\cdot)\leq T\F(A)$,
and for $0\leq t\leq u\leq t+T$, 
$$
\EE^{(t,\gamma)}\left[\exp\left(2 \int_t^u \|g(u-s)\|^2 \Gamma_s ds\right)\right]
 = \exp\Big\{\gamma\psi(u-t) + \phi(u-t)\Big\}.
$$
\end{proposition}
Here and after, we use the shorthand notation~$\EE^{(t,\gamma)}$ 
to denote expectation conditional on the event $\{\Gamma_t=\gamma\}$.
Likewise, the notation~$\Gamma^{(t,\gamma)}$ shall mean the process~$\Gamma$
started at~$\gamma$ at time~$t$.
\begin{proof}
We first show existence of solutions. 
With the function
$\varphi(t) := \psi(t) - 2G(t)$, the ODE in the proposition is equivalent to the following:
\begin{align}
\partial_t \varphi(t) = \R(\varphi(t) + 2G(t)).\label{barpsi}
\end{align}
Consider now the modified ODE
$$
\partial_t u(t) = \R((u(t)+ 2G(t))\wedge A).
$$
Since the right-hand side is Lipschitz in~$u$, by the
Cauchy-Lipschitz theorem, this equation admits a unique solution,
denoted by $\varphi^{A}(t)$. 
Since $\R(0)=0$, this solution is nonnegative. 
Moreover, in view of the convexity of~$\R$, it can be bounded from above as follows:
$$
\max_{0\leq t\leq T}\varphi^{A}(t) \leq T \max_{0\leq t\leq T}
\R\left(({\varphi^{A}(t)} + 2G(t))\wedge A\right)\leq T (0\vee \R(A)).  
$$
If $T (0\vee \R(A)) + 2G(T) \leq A$, 
then this solution coincides on $[0,T]$ with the solution of the original equation~\eqref{barpsi},
which therefore exists and is bounded from above. 

Let us now prove the formula for the Laplace transform. Consider the process
$$
M_s := \exp\left\{2\int_t^s \|g(u-r)\|^2 \Gamma^{(t,\gamma)}_r dr\right\} 
\exp\left\{\Gamma^{(t,\gamma)}_s \psi(u-s) +\phi(u-s)\right\},
\quad \text{for } t\leq s\leq u\leq T+t.
$$
By It\^o's formula,
$$
\frac{dM_s}{M_{s-}} = \psi(u-s) d\Gamma^c_s + \int_0^\infty
\left(e^{z\psi(u-s)}-1\right) \widetilde J_\Gamma (ds\times dz),
$$
where $\Gamma^c$ denotes the continuous martingale part of $\Gamma$, 
and $\widetilde J_\Gamma$ its compensated jump measure.
This means that $M$ is a local martingale. The process
$$
Z_s := \int_t^s \psi(u-r) d\Gamma^c_r + \int_t^s \int_0^\infty
\left(e^{z\psi(u-r)}-1\right) \widetilde J_\Gamma (dr\times dz),
$$
together with $\Gamma$, is a two-dimensional time-inhomogeneous affine process, 
which satisfies the conditions of~\cite[Theorem 3.1]{kallsen2010exponentially}, 
so that {its stochastic exponential, and therefore the process $M$,}
is a true martingale on $[t,u]$. 
Taking expectations conditional on $\{\Gamma_t=\gamma\}$, we obtain
$$
\exp\Big\{\gamma \psi(u-t) + \phi(u-t)\Big\}
 = \EE\left[\exp\left\{2\int_t^u \|g(u-s)\|^2 \Gamma^{(t,\gamma)}_s ds\right\}\right]. 
$$
\end{proof}
As in Section~\ref{toy.sec}, it is more straightforward to deal with the forward variance
$\xi_t(u) := \EE[\sigma_u^2|\mathcal F_t]$ than with the instantaneous
volatility. 
The following result follows by conditioning (for the first part) and by an application of It\^o's formula:
\begin{proposition}\label{prop:DynamicsFwdVar}
Under Assumption~\ref{assu:Levy}, the forward variance process is given by 
\begin{align*}
\xi_t(u) = \xi_0(u)\exp\left(2\int_{0}^t \sqrt{\Gamma_s} g(u-s)^\top
  dW_s+\psi(u-t) \Gamma_t + \phi(u-t) - \psi(u)\Gamma_0 - \phi(u)\right),
\end{align*}
for $0\leq t\leq u \leq T$, and its dynamics reads
$$
d\xi_t(u) = 2\xi_t(u)  \sqrt{\Gamma_t} g(u-t)^\top dW_t+ \xi_t(u)\psi(u-t)d\Gamma^c_t
 + \xi_{t-}(u)\int_{\RR}\left(e^{\psi(u-t) z}-1\right)\widetilde J_\Gamma(dt\times dz).
$$
\end{proposition}
Similarly to the Gaussian Volterra case, the process~$X$ is non-Markovian, 
but the forward variance curve $(\xi_t(u))_{u\geq t}$ together with $\Gamma_t$ is Markovian and the current state of the forward variance curve, which is observable from option prices, and of $\Gamma_t$ determines the future dynamics. 
Mimicking Section~\ref{volterra.sec}, 
the value at time~$t$ of a Call option on the VIX is given by 
$$
P_t = \EE\left[f\left(\frac{1}{\Theta}\int_{T}^{T+\Theta} \xi_{T}(u) du\right) \Big|\mathcal  F_t\right]
 = F_{\Gamma}(t,\Gamma_t,\xi_t(u)_{T\leq u \leq T+\Theta}),
$$
with $f(x)=(\sqrt{x}-K)_+$, 
where $F_{\Gamma}$ is the deterministic map from $[0,T]\times\RR_+\times \Hh$ to $\RR$ defined by 
\begin{equation}\label{vixasianGamma}
F_{\Gamma}(t, \gamma, x) := \EE\left[f\left(\frac{1}{\Theta}\int_{T}^{T+\Theta}x(u) \Ee_{t,{T}}(\gamma,u) du\right) \right], \end{equation}
with
$$
\Ee_{t,T}(\gamma,u):= \exp\left(2\int_t^T \sqrt{\Gamma^{(t,\gamma)}_s} g(u-s)^\top dW_s
 + \psi(u-T)\Gamma^{(t,\gamma)}_T + \phi(u-T)-\psi(u-t)\gamma-\phi(u-t)\right).
$$

\begin{example}\label{ou.ex}
Let $\Gamma$ be a  L\'evy-driven positive Ornstein-Uhlenbeck process satisfying
$$
d\Gamma_t = -\lambda \Gamma_t dt + dL_t,\qquad \text{with }
\EE[e^{u L_t}] = \exp\Big(\Psi(u) t\Big),
$$
so that $\R(u) = -\lambda u$ and $\F(u) = \Psi(u)$. 
The assumptions of Proposition~\ref{affine.prop} are satisfied for any~$T$ such that $\Psi(2G(T))$ is finite, and in this case,
$$
\phi(t) :=\int_0^t \Psi\left(\psi(s)\right) ds
\qquad\text{and}\qquad 
\psi(t) =2\int_0^t e^{-\lambda(t-s)} \|g(s)\|^2 ds.
$$
From Proposition~\ref{prop:DynamicsFwdVar}, the dynamics of the forward variance process is therefore given by
$$
\frac{d\xi_t(u)}{\xi_t(u)} = 2  \sqrt{\Gamma_t} g(u-t)^\top dW_t+
\int_{\RR_+} \left\{\exp\left(2z\int_t^u e^{-\lambda(s-t)} \|g(u-s)\|^2ds\right)-1\right\}
\widetilde J_L(dt\times dz),
$$
where~$\widetilde J_L$ is the compensated jump measure of~$L$. 
Assume for example that~$L$ has jump intensity $\Lambda>0$ and exponential jump size distribution with parameter $a>0$, then 
$$
\Psi(u) = a \Lambda \int_0^\infty (e^{ux}-1) e^{-ax} dx =
\Lambda\left(\frac{a}{a-u}-1\right) = \frac{\Lambda u} {a-u},
\qquad\text{for all }u<a.
$$
\end{example}

\begin{example}\label{cirex}
Let $\Gamma$ be the CIR process with dynamics
$$
d\Gamma_t = k(\theta-\Gamma_t) dt + \delta \sqrt{\Gamma_t} dB_t,
$$
where $B$ is a standard Brownian motion independent of $W$. Then,
$\F(u) = k\theta u$ and $\R(u) = - k u + \frac{\delta^2}{2}u^2$.
Assumption \ref{assu:Levy} can be shown to be satisfied for any $T$ such that $4 G(T)
T \delta^2\leq 1$. Under this assumption, the dynamics of the forward
variance process is given by
$$
d\xi_t(u) = 2\xi_t(u)  \sqrt{\Gamma_t} g(u-t)^\top dW_t+
\xi_t(u)\psi(u-t)\delta \sqrt{\Gamma_t} dB_t,
$$
where $\psi$ is a deterministic function defined in Proposition~\ref{affine.prop}. 
In this case, $\xi(u)$ is a Heston process with uncorrelated stochastic volatility, 
so that, contrary to the jump case where the skew arises from downward jumps, 
we have here a symmetric implied volatility smile.
\end{example}
\paragraph{Martingale representation and hedging of VIX options}{
The martingale representation result in the presence of the extra risk
source is more involved and we provide it without proof, assuming
sufficient regularity to apply the It\^o formula.  Since the jump part of $\Gamma$ has finite variation, we may apply
the simplified version of It\^o formula which gives
\begin{align*}
dP_t &= \int_T^{T+\Theta}D_x F_\Gamma(t,\Gamma_t,
\xi_{t})(u) \xi_{t}(u)\left\{2\sqrt{\Gamma_t}g(u-t)^\top dW_t + \psi(u-t)
  d\Gamma^c_t \right\}du+ D_\Gamma
F_\Gamma(t,\Gamma_t,\xi_{t})d\Gamma^c_t \\
&\int_{\mathbb R}\left[
  F_{\Gamma}\left(t,\Gamma_{t-}+z,(\xi_{t-}e^{\psi(u-t)z})_{T\leq u\leq T+\Theta}\right)
 - F_{\Gamma}(t,\Gamma_{t-},\xi_{t-})\right]\widetilde{J}_\Gamma(dt\times dz).
\end{align*}
For the hedging portfolio with value 
$V_t = P_0+\int_0^T \int_T^{T+\Theta} \eta_s(u) d\xi_s(u)$, we have
$\EE^{\QQ}[(P_T - V_T)^2] = \EE\left[\int_0^T d\langle P-V\rangle_t\right]$,
where
\begin{align*}
&d\langle P-V\rangle_t = \left\langle P_t-\int_T^{T+\Theta} \eta_t(u)d\xi_t(u) \right\rangle_t\\
&  = 4 \Gamma_t \left\{\int_T^{T+\Theta}\left[D_x F_\Gamma(t,\Gamma_t,
\xi_{t})(u) - \eta_t(u)\right]\xi_{t}(u)g(u-t)^\top  du \right\}^2 dt \\
&+ \frac{\delta^2 \Gamma_t}{2}\left\{\int_T^{T+\Theta}
\Big(D_x F_\Gamma(t,\Gamma_t,
\xi_{t})(u) \xi_{t}(u)\psi (u-t)+ D_\Gamma F_\Gamma(t,\Gamma_t,
\xi_{t})  - \eta_t(u)\xi_t(u)\psi(u-t)\Big)  du \right\}^2 dt\\
& + \int_{\RR} (\mathrm{m}(dz) + \Gamma_t
  \mu(dz))\Big\{F_{\Gamma}(t,\Gamma_{t}+z,(\xi_{t}e^{\psi(u-t)z})_{T\leq
  u\leq T+\Theta})-F_{\Gamma}(t,\Gamma_{t},\xi_{t}) \\ &-
  \int_T^{T+\Theta} \eta_t(u) \xi_t(u) (e^{\psi(u-t)z}-1)du\Big\}^2dt.
\end{align*}
For perfect hedging, the following system must therefore hold:
\begin{equation*}
\left\{\begin{array}{rl}
\displaystyle\int_T^{T+\Theta}(D_x F_\Gamma(t,\Gamma_t,
\xi_{t})(u) - \eta_t(u))\xi_{t}(u)g(u-t)^\top  du & = 0,\\
\displaystyle\int_T^{T+\Theta}(D_x F_\Gamma(t,\Gamma_t,
\xi_{t})(u) \xi_{t}(u)\psi (u-t)+ D_\Gamma F_\Gamma(t,\Gamma_t,
\xi_{t})  - \eta_t(u)\xi_t(u)\psi(u-t))  du & = 0,\\
\displaystyle F_{\Gamma}(t,\Gamma_{t}+z,(\xi_{t}e^{\psi(u-t)z})_{T\leq
  u\leq T+\Theta})-F_{\Gamma}(t,\Gamma_{t},\xi_{t}) -
  \int_T^{T+\Theta} \eta_t(u) \xi_t(u) (e^{\psi(u-t)z}-1)du & = 0,
\end{array}
\right.
\end{equation*}
where the last one is to hold for all $z$ on the support of $\mathrm m + \mu$. 
Assume first that the process $\Gamma$ has no jump part. Then, only
the first two equations must be solved: this is possible whenever one
can find numbers $u_1,\dots,u_{d+1}\in [T,T+\Theta]$ such that the
vectors $(g(u_j-t)^1,\dots, g(u_j-t)^d,\psi(u_j-t))$ for
$j=1,\dots,d+1$ are linearly independent for all $t\in [0,T]$. In this
case, despite the presence of stochastic volatility, a VIX option may
be perfectly hedged using only forward variance curve: in the interest
rate language, there is no \emph{unspanned stochastic volatility}. 

In the presence of jumps, the hedging problem is more complex. When
the support of $\mathrm m + \mu$ contains only a finite number of
points, $z_1,\dots,z_k$, perfect hedging is possible whenever one
can find numbers $u_1,\dots,u_{d+k+1}\in [T,T+\Theta]$ such that the
vectors $(g(u_j-t)^1,\dots,
g(u_j-t)^d,\psi(u_j-t),e^{\psi(u_j-t)z_1},\dots, e^{\psi(u_j-t)z_k})$
are linearly independent for all $t\in [0,T]$. However, the hedge
ratios will be unstable when the corresponding matrix is close to
being singular. The method may therefore work when the number of
points in the support of $\mathrm m + \mu$ is small, but will be
difficult to implement for a large number of points, and a fortiori
when trying to approximate a continuous jump size distribution with a
discrete one. 

\begin{example}
In the CIR case (Example~\ref{cirex}), 
the VIX option price has the martingale representation
\begin{align*}
dP_t &= \left\{2\int_T^{T+\Theta}D_x F_\Gamma(t,\Gamma_t,
\xi_{t})(u) \xi_{t}(u) g(u-t)^\top  du\right\} \sqrt{\Gamma_t} dW_t \\
& +\left\{ D_\Gamma
F_\Gamma(t,\Gamma_t,\xi_{t})+\int_T^{T+\Theta}D_x F_\Gamma(t,\Gamma_t,
\xi_{t})(u) \xi_{t}(u) \psi(u-t)du
  \right\} \delta \sqrt{\Gamma_t} dB_t.
\end{align*}
On the other hand, considering the continuous version of the variance swap, 
the dynamics of the forward variance swap between~$T$ and~$T+\Theta$ is given by 
$$
dS^{T,\Theta}_t = \left\{\frac{2}{\Theta}\int_T^{T+\Theta}\xi_t(u)
g(u-t)^\top du\right\} \sqrt{\Gamma_t} dW_t+
\left\{\frac{\delta}{\Theta}\int_T^{T+\Theta}\xi_t(u)\psi(u-t)du\right\} \sqrt{\Gamma_t} dB_t.
$$
It is thus clear that we can construct a portfolio of two variance swaps with
different maturities which will perfectly offset the risk of a VIX
option. 
\end{example}
}

\subsection{Pricing VIX options by Monte Carlo}\label{montecarlo2.sec}
We extend here the numerical analysis from Section~\ref{montecarlo.sec}
to the modulated Volterra case, 
essentially based on deriving an approximation for~\eqref{vixasianGamma}. 
The two discretisation schemes are adapted as follows:
\begin{itemize}
\item the rectangle scheme (with $\zeta^n_i$ and $\eta^n(u)$ defined as before):
$$
F_n(t,\gamma,x)
 :=\EE\left[f\left(\frac{1}{\Theta}\sum_{i=0}^{n-1}\zeta^n_i  \Ee_{t,T}(\gamma,t^n_i)\right) \right] = \EE\left[f\left(\frac{1}{\Theta}\int_{T}^{T+\Theta}x(u) \Ee_{t,T}(\gamma,\eta^n(u)) du\right) \right];
$$
\item the trapezoidal scheme:
\begin{align*}
\widehat F_n(t,\gamma,x)
 & := \EE\left[f\left(\frac{1}{\Theta}\sum_{i=0}^{n-1}\int_{t^n_i}^{t^n_{i+1}} x(u) (
    \theta^n (u)\Ee_{t,T}(\gamma,t^n_i) + (1-\theta^n(u)) \Ee_{t,T}(\gamma,t^n_{i+1}) )du\right) \right]\\
& = \EE\left[f\left(\frac{1}{\Theta}\int_{T}^{T+\Theta} x(u) (
    \theta^n (u)\Ee_{t,T}(\gamma,\eta^n(u)) + (1-\theta^n(u)) \Ee_{t,T}(\gamma,\overline{\eta}^n(u)) )du\right) \right].
\end{align*}
Similarly to the previous case, the sequence of random variables $(Z_i)_{i=0}^{n-1}$ with 
$$
Z_{i}:= \log \Ee_{t,T}(\gamma,t_i)
 = 2\int_{t}^T \sqrt{\Gamma^{(t,\gamma)}_s} g(t_i-s)^\top  dW_s  + m_i,
$$
forms a conditionally Gaussian random vector (when conditioning on the trajectory $\left(\Gamma_{s}\right)_{t\leq s \leq T}$, given $\Gamma_t=\gamma$) with mean~$m_i$ 
and covariance~$C_{i,j}$ given by
\begin{align*}
m_i & = \psi(t_i-T) \Gamma^{(t,\gamma)}_T + \phi(t_i-T) - \psi(t_i-t)\gamma - \phi(t_i-t),\\
C_{ij} & = 4\int_t^T \Gamma_s g(t^n_i-s)^\top g(t^n_j-s) ds.
\end{align*}
\end{itemize}
The Monte Carlo computation is then implemented in two successive steps:
\begin{itemize}
\item For the covariances $C_{ij}$: 
having simulated a trajectory $\left(\Gamma_{s}\right)_{t\leq s \leq T}$ 
we compute the conditional covariances~$C_{ij}$ given~$\Gamma$. 
When~$\Gamma$ is a L\'evy-driven OU process with finite jump intensity, this simulation does not induce a discretisation error, since the integral describing~$C_{i,j}$ is in fact a finite sum over the jumps of the L\'evy process. Otherwise, we need to simulate a discretised trajectory of~$\Gamma$, but this simulation is fast, since~$\Gamma$ is Markovian (see below). 
\item Simulate the random Gaussian vector $(Z_i)_{i=0}^{n-1}$ and compute the option pay-off. 
\end{itemize} 

The following proposition extends Proposition~\ref{mc.prop} 
and characterises the convergence rate of these two discretisation schemes, 
assuming that the simulation of $C_{ij}$ is done without error. 
\begin{proposition}\label{mc2.prop}
Let Assumption~\ref{assu:Levy} hold.
Assume further that $f$ is Lipschitz, $x(\cdot)$ bounded, and that there exist
$c<\infty$ and $\beta>0$ such that for all $T\leq t_1 <t_2$, 
$$
\left(\int_t^T \|g(t_2-s)-g(t_1-s)\|^2ds\right)^{1/2} \leq c(t_2-t_1) (t_2-T)^{\beta-1}. 
$$
\begin{itemize}
\item
If 
\begin{equation}\label{expmoment}
\EE\left[\exp\left(8\int_t^u \Gamma^{(t,\gamma)}_s \|g(u-s)\|^2ds\right)\right]
\end{equation} 
is bounded uniformly for $u\in[T,T+\Theta]$, and 
$\EE[(\Gamma^{(t,\gamma)}_s)^2]$ is bounded uniformly for $s\in [t,T]$,
then on~$\Tt_1$, 
$|F(t,x)-F_n(t,x)|= \mathcal{O}\left(\frac{1}{n}\right)$;
\item if, in addition to the above assumptions, $\EE[(\Gamma^{(t,\gamma)}_s)^4]$ is bounded uniformly for $s\in [t,T]$,
$\|g\|$ is positive and decreasing, and there exists $c<\infty$ such that for all 
$T\leq t_1 \leq t_2<t_3$,
\begin{align*}
&\|g(t_1-T)\|^2 \leq c (t_1-T)^{\beta-1},\\
&\|g(t_1-T)\|^2 - \|g(t_2-T)\|^2 \leq C(t_2-t_1)(t_1-T)^{\beta-2},\\
&\left(\int_t^T \|g(t_2-s)-\frac{t_3-t_2}{t_3-t_1}g(t_1-s)- \frac{t_2-t_1}{t_3-t_1}g(t_3-s)\|^2ds\right)^{1/2}
\leq c(t_3-t_1)^2 (t_3-T)^{\beta-2}. 
\end{align*}
Then on~$\Tt_{\kappa}$ with $\kappa(\beta+1)>2$, 
$|F(t,x)-\widehat F_n(t,x)|= \mathcal{O}\left(\frac{1}{n^2}\right)$.
\end{itemize}
\end{proposition}
\begin{remark}
Similarly to Proposition~\ref{affine.prop}, it can be shown that a
sufficient condition for~\eqref{expmoment} to be bounded uniformly on
$u\in[T,T+\Theta]$ is that $\int_1^\infty z e^{zA}\{\mu(dz) + \m(dz)\}$ is finite and, for some $A>0$,
$$
8G(T+\Theta-t) + (T+\Theta-t)(0\vee \R(A))<A.
$$
Moreover, the exponential integrability of jump sizes ensures that $\EE[(\Gamma^{(t,\gamma)}_s)^4]$
is uniformly bounded. 
\end{remark}

\subsection{Approximate pricing and control variate}\label{approx.sec}
As in the Gaussian Volterra case, we can obtain a simple approximate
formula for the VIX option price by replacing the integral over a family
of conditionally lognormal random variables by the exponential of
the integral of their logarithms. 
In other words, we replace the squared VIX index~\eqref{eq:VIXFormula}
by the approximation~\eqref{eq:ApproxVIXSquared}.
From the expression of the forward variance in Proposition~\ref{prop:DynamicsFwdVar}, we have
\begin{align*}
\log \overline{\VIX}^2_T &= \frac{1}{\Theta}\int_T^{T+\Theta} \log \xi_t(u) du + 2\int_t^T
  \Gamma^{(t,\gamma)}_s  dW_s \frac{1}{\Theta}\int_T^{T+\Theta} g(u-s)^\top
  du 
  \\ &+\frac{\Gamma^{(t,\gamma)}_T}{\Theta}\int_T^{T+\Theta}\psi(u-T)du
  + \frac{1}{\Theta}\int_T^{T+\Theta}(\phi(u-T)-\phi(u-t))
  du-\frac{\gamma}{\Theta}\int_T^{T+\Theta}\psi(u-t)
  du
\end{align*}
The characteristic exponent of $\log \overline{\VIX}^2_T$ is therefore given by 
\begin{align*}
\Psi(z) & := \log\EE\left[\left.e^{iz \log \overline{\VIX}^2_T}\right|\mathcal F_t\right]\\
 & = 
\frac{iz}{\Theta}\int_T^{T+\Theta} \log \xi_t(u) du
 + \frac{iz}{\Theta}\int_T^{T+\Theta}(\phi(u-T)-\phi(u-t))du
 - \frac{iz\gamma}{\Theta}\int_T^{T+\Theta}\psi(u-t) du\\
& + \log \EE\left[\exp\left(-2 z^2 \int_t^T\Gamma_s^{(t,\gamma)} G^2(s,T,\Theta) ds
 + \frac{iz \Gamma^{(t,\gamma)}_T}{\Theta}\int_T^{T+\Theta}\psi(u-T)du\right)\right],
\end{align*}
where $G(s,T,\Theta ) := \frac{1}{\Theta}\int_{T}^{T+\Theta} g(u-s)du$.
Similarly to Proposition~\ref{affine.prop}, under appropriate
integrability conditions this expectation can be reduced to ordinary differential equations.
In the context of the L\'evy-driven OU process of Example~\ref{ou.ex}, 
the computation can be done explicitly. In this case, 
$$
\Gamma^{(t,\gamma)}_s = \gamma e^{-\lambda(s-t)} + \int_t^s
e^{-\lambda(s-r)} dL_r,
$$
so that 
\begin{align*}
 &\EE\left[\exp\left(-2 z^2 \int_t^T  \Gamma_s^{(t,\gamma)} G^2(s,T,\Theta) ds
  + \frac{iz \Gamma^{(t,\gamma)}_T}{\Theta}\int_T^{T+\Theta}\psi(u-T)du\right)\Big|\mathcal F_t\right]\\
& = \exp\left(-2 z^2 \int_t^T \gamma e^{-\lambda(s-t)}  G^2(s,T,\Theta) ds
 + \frac{iz \gamma e^{-\lambda(T-t)} }{\Theta}\int_T^{T+\Theta}\psi(u-T)du\right)\\
& \times \EE\left[\exp\left(-2 z^2 \int_t^T dL_r \int_r^Te^{-\lambda(s-r)}  G^2(s,T,\Theta) ds
 + \frac{iz \int_t^T e^{-\lambda(T-r)} dL_r}{\Theta}\int_T^{T+\Theta}\psi(u-T)du\right)\right]\\
& = \exp\left(-2 z^2 \int_t^T \gamma e^{-\lambda(s-t)}  G^2(s,T,\Theta) ds
 + \frac{iz \gamma e^{-\lambda(T-t)} }{\Theta}\int_T^{T+\Theta}\psi(u-T)du\right)\\
&\times \exp\left(\int_t^T\Psi\left(-2 z^2 \int_r^T e^{-\lambda(s-r)}  G^2(s,T,\Theta) ds
 + \frac{iz e^{-\lambda(T-r)}}{\Theta}\int_T^{T+\Theta}\psi(u-T)du\right)dr\right)
\end{align*}

\subsection{Numerical illustration}
We consider here the volatility-modulated model~\eqref{modulated.eq}
with power law kernel $g(t,s) = (t-s)^{H-\frac{1}{2}}$,
and where $\Gamma$ is the L\'evy-driven OU process 
from Example~\ref{ou.ex}. 
Using the Monte Carlo method in Section~\ref{montecarlo2.sec},
Proposition~\ref{mc2.prop} holds provided $\Psi(8G(T+\Theta-t))<\infty$, 
so that the convergence rate is $n^{-1}$ for the rectangle scheme, 
and $n^{-2}$ for the trapezoidal scheme with appropriate discretisation grids. 
The driving L\'evy process~$L$ 
is a compound Poisson process with one-sided exponential jumps\footnote{Empirical tests suggest that double-sided jumps do not significantly improve the calibration.}; 
the model parameters are $\lambda$ (mean reversion of the OU process), 
$\Lambda$ (jump intensity), $a$ (parameter of the exponential law), 
$\gamma$ (initial value of the OU process), 
$\xi_0(T)$ (initial forward variance curve), and~$H$ 
(fixed to the value $0.1$ in accordance with the findings in~\cite{gatheral2014volatility}). 
Figure~\ref{im:Figureapproxaccurate} indicates that the approximation formula in Section~\ref{approx.sec} is very accurate;
compared to Monte Carlo schemes with $90$ approximation steps (Ndisc on the plots),
our approximation has the benefit of being much faster to compute. 
This convergence is illustrated with the following parameters:
maturity is one month, and 
$(\lambda, \Lambda, a, \gamma, \xi_0(T)) = (0.08, 0.71, 6.18, 0.05, 0.013)$
(which corresponds to a set of calibrated parameters below).

\begin{figure}[h!]
\begin{center}\includegraphics[scale=0.35]{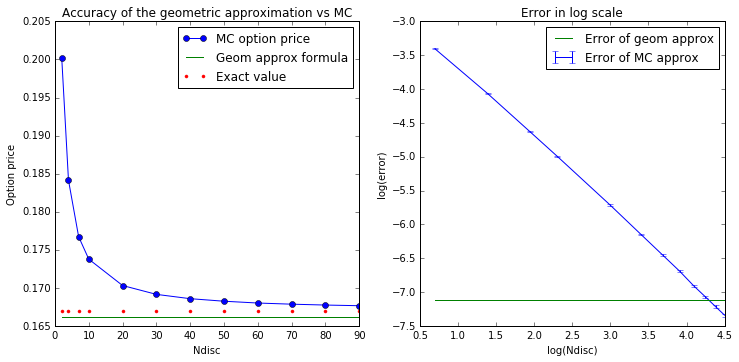}\end{center}
\caption{Accuracy of the approximation formula as compared to the Monte Carlo scheme.}
\label{im:Figureapproxaccurate}
\end{figure}


We calibrate the model to VIX options on May 14, 2014, for five maturities,
using the approximate pricing formula of Section~\ref{approx.sec}, 
by minimising the sum of squared differences between market prices and model prices, 
using the L-BFGS-B algorithm (Python \texttt{optimize} toolbox). 

\subsubsection{Slice by slice calibration}
In this test each maturity has been calibrated separately, 
and the forward variance value~$\xi_0(T)$ for each maturity 
has also been calibrated to VIX options.
The calibration results are shown in Figure~\ref{calib.fig}, and the calibrated 
parameters in Table~\ref{calib.tab}.
The error is the square root of the mean square error of option prices (in USD).
The calibration time on a standard PC ranges from 20 to 100 seconds,
depending on the starting value of parameters. 
The calibration quality shows an overall error below~$3$ cents, 
and the parameters appear to be reasonably stable over all maturities.

\begin{table}[h!]
\centerline{\begin{tabular}{c|ccccc|c}
Maturity (in days) & $\lambda$ & $\Lambda$ & $a$ & $\gamma$ & $\xi_0(T)$  & Error\\
\hline
7 & 0.08 & 0.71 & 6.18 & 0.05 & 0.013 & 0.005 \\ 
35 & 0.01 & 5.82 & 19.81 & 0.007 & 0.014 & 0.03 \\ 
63 & 0.01& 6.61 & 25.41 & 0.01 & 0.012 & 0.16 \\ 
98 & 0.01& 5.63 & 28.70 & 0.001 & 0.012 & 0.008 \\ 
126 &  0.92 & 4.97 & 25.19 & 0.001 & 0.011 & 0.023 \\ 
\hline  
\end{tabular}}
\caption{Calibrated parameter values corresponding to Figure~\ref{calib.fig}.}
\label{calib.tab}
\end{table}

\begin{figure}
\centerline{\includegraphics[scale=0.4]{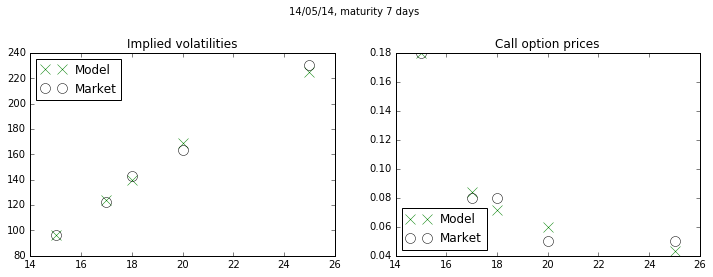}}
\centerline{\includegraphics[scale=0.4]{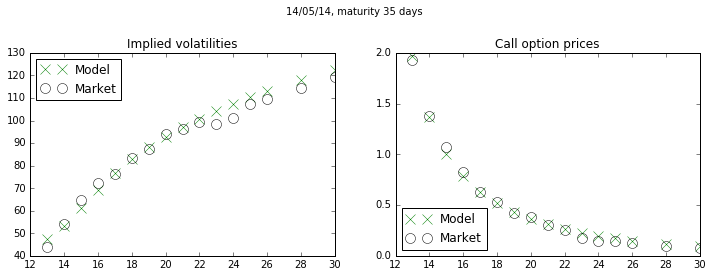}}
\centerline{\includegraphics[scale=0.4]{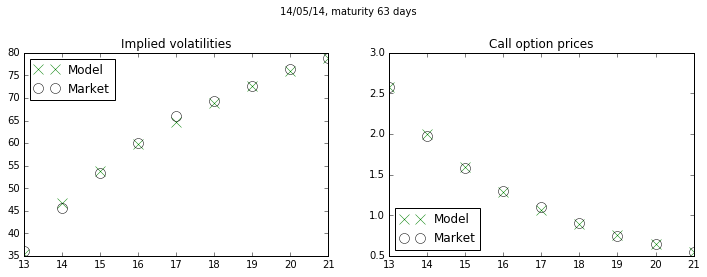}}
\centerline{\includegraphics[scale=0.4]{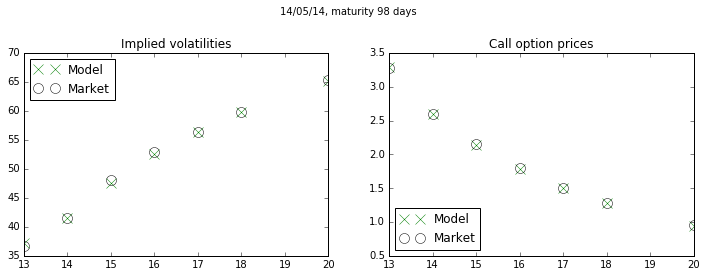}}
\centerline{\includegraphics[scale=0.4]{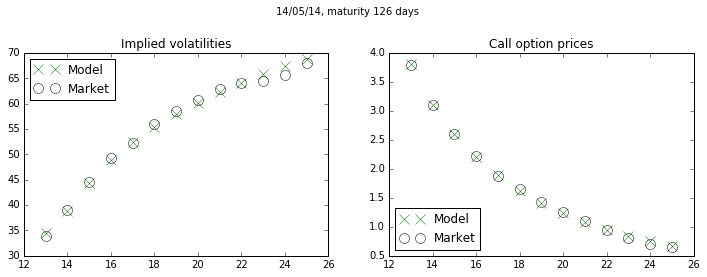}}
\caption{Option prices and implied volatilities of VIX options observed in the market 
calibrated (slice by slice) from the model with one-sided exponential jumps. 
From top to bottom, the maturities are $7$, $35$, $63$, $98$, and $126$ days. 
The forward variance is calibrated to VIX option prices.}
\label{calib.fig}
\end{figure}

\subsubsection{Slice by slice calibration with pre-specified forward variance curve}
We now consider a joint calibration procedure, where each maturity is calibrated separately, 
but the forward variance~$\xi_0(\cdot)$ is 
computed from SPX option prices using the VIX replication formula. 
Figure~\ref{calib2.fig} shows the results of the calibration, 
with calibrated parameters in Table~\ref{calib2.tab}.
The pricing errors are now slightly larger, yet still acceptable. 

\begin{table}[h!]
\centerline{\begin{tabular}{c|ccccc|c}
Maturity (in days) & $\lambda$ & $\Lambda$ & $a$ & $\gamma$ & $\xi_0(T)$  & Error \\
\hline
7 & 0.086 & 0.583 & 5.410 &  0.272 & 0.013 & 0.013 \\ 
35 & 0.008 & 0.597 & 9.903 &  0.088 & 0.018 & 0.041 \\ 
63 & 0.01 &  0.08 &  15.24 &  0.13 & 0.022 & 0.066 \\ 
98 & 0.009 & 0.06 &  0.028 &  0.11 & 0.027 & 0.095 \\ 
126 &  0.922 & 0.094 & 0.001 &  0.149 & 0.030 & 0.074 \\ 
\hline
\end{tabular}}
\caption{Calibrated parameter values corresponding to Figure~\ref{calib2.fig}.}
\label{calib2.tab}
\end{table}

\begin{figure}
\centerline{\includegraphics[scale=0.4]{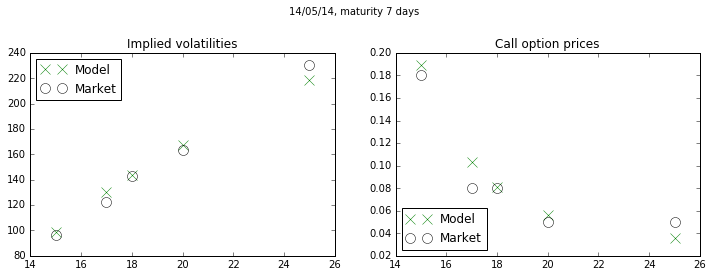}}
\centerline{\includegraphics[scale=0.4]{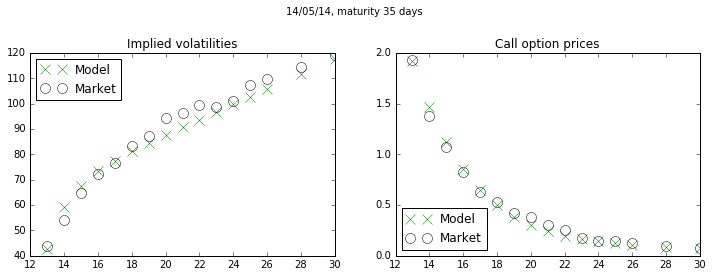}}
\centerline{\includegraphics[scale=0.4]{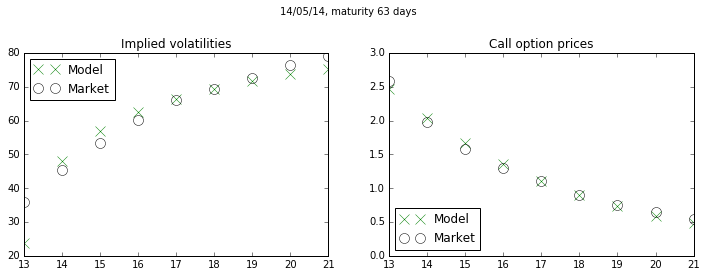}}
\centerline{\includegraphics[scale=0.4]{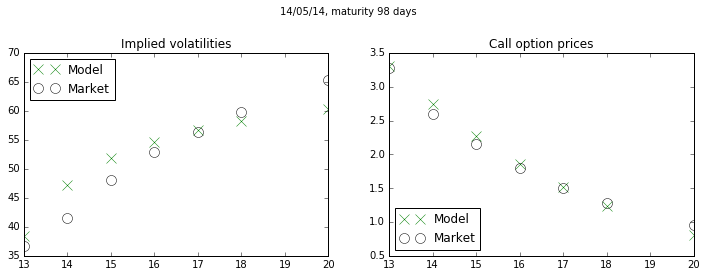}}
\centerline{\includegraphics[scale=0.4]{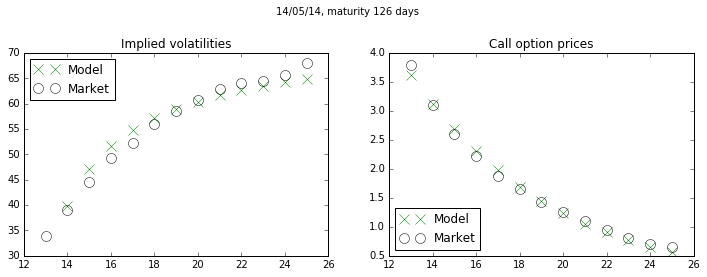}}
\caption{Joint calibration on VIX Smile and term structure of the SPX implied volatility.}
\label{calib2.fig}
\end{figure}

\newpage
\subsubsection{Joint calibration to several maturities}
We finally test the calibration over several maturities at the same time.
Figure~\ref{calib3.fig} shows the result of the simultaneous calibration to three maturities ($35$, $63$ and~$98$ days), 
where the forward variance is calibrated separately from SPX option prices as in Table~\ref{calib2.tab}. 
The optimal parameters and errors are
$(\lambda, \Lambda, a, \gamma) = (0.29771, 0.915, 9.576, 0.028)$, and
\begin{table}[h!]
\centerline{\begin{tabular}{c|ccccc|c}
\hline
Maturity (in days) & $35$ & $63$ & $98$ \\
\hline
Error  & 0.084 & 0.12 & 0.15 \\
\hline
\end{tabular}}
\label{calib3.tab}
\end{table}

\begin{figure}[h!]
\centerline{\includegraphics[scale=0.4]{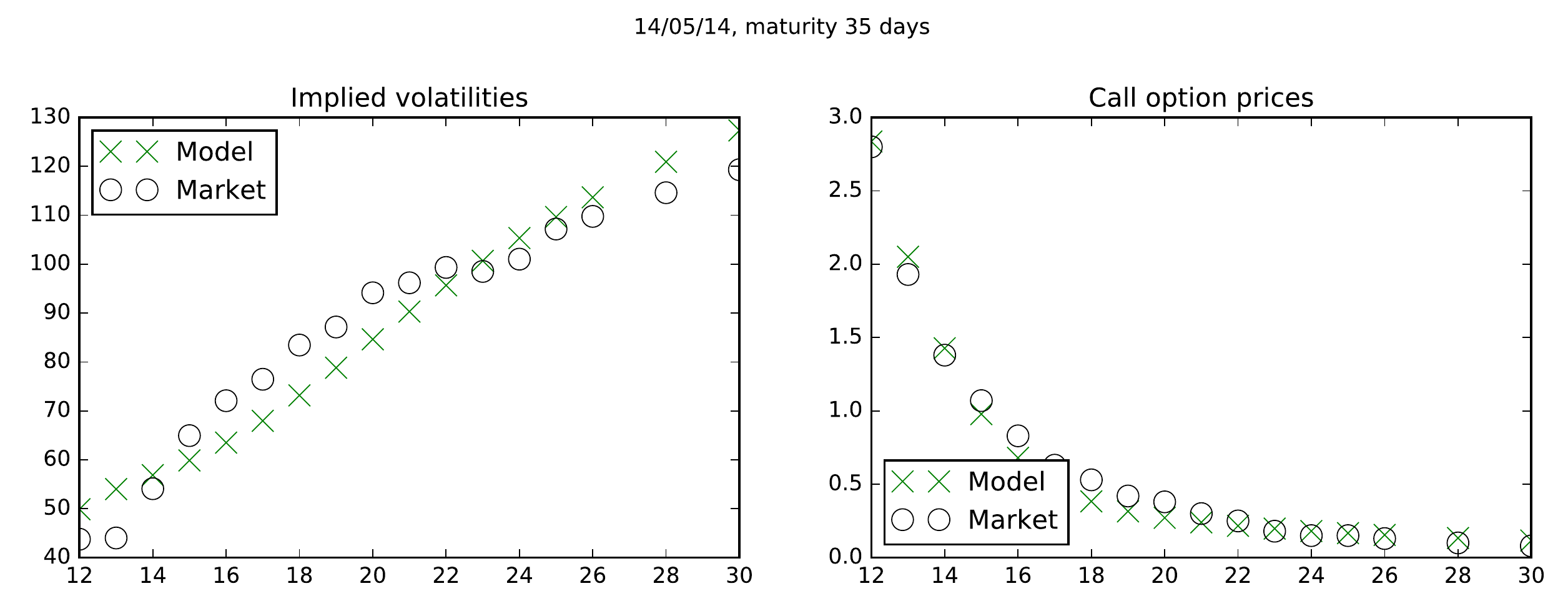}}
\centerline{\includegraphics[scale=0.4]{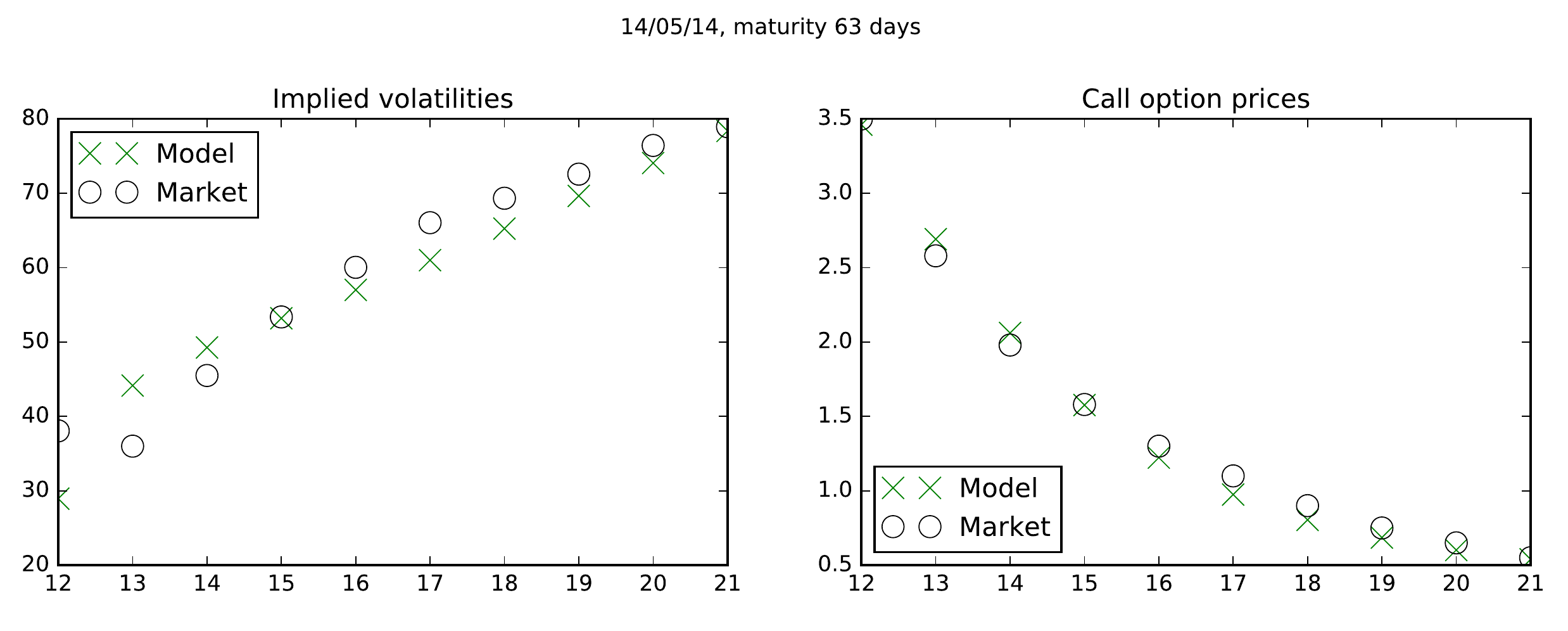}}
\centerline{\includegraphics[scale=0.4]{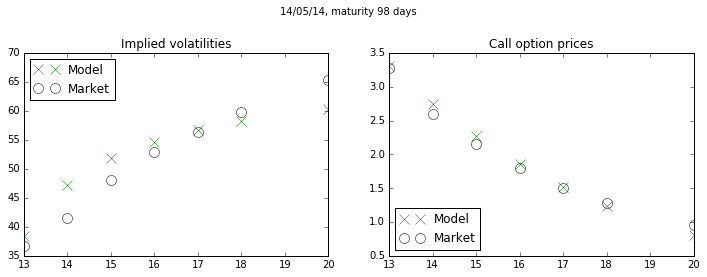}}
\caption{Option prices and implied volatilities of VIX options,
calibrated simultaneously to all three maturities.
The graphs represent the maturities $35$,
$63$ and $98$ days.}
\label{calib3.fig}
\end{figure}



\appendix

\small
\section{Appendix}

\subsection{Proof of Corollary~\ref{mcrb1.corr}}
We consider $\alpha=1$ with no loss of  generality, and assume that $H\leq \frac{1}{2}$, 
the case where $H>\frac{1}{2}$ being similar. 
The proof relies on checking the assumptions of Proposition~\ref{mc.prop}. 
Let us start with the rectangle scheme. 
On the one hand, when $t_2-T \geq 2(t_2-t_1)$, then also $t_1-T\leq \frac{t_2-T}{2} $, 
and the following estimate holds true:
\begin{align*}
&\left(\int_t^T \|g(t_2,s)-g(t_1,s)\|^2ds\right)^{1/2} = \left[\int_t^{T}
  \left(|t_{2}-s|^{H-\frac{1}{2}}-|t_1-s|^{H-\frac{1}{2}}\right)^2
  ds\right]^{1/2}\\
&\leq C \left[\int_t^{T}
  (t_{2}-t_1)^2 (t_1-s)^{2H-3}
  ds\right]^{1/2}
\leq C (t_{2}-t_1) (t_{1}-{T})^{H-1}\leq C2^{1-H} (t_{2}-t_1)
  (t_{2}-T)^{H-1}.
\end{align*}
On the other hand, when $t_2-T < 2(t_2-t_1)$, we have the bound
\begin{align*}
&\left(\int_t^T \|g(t_2,s)-g(t_1,s)\|^2ds\right)^{1/2} 
 \leq \left(\int_0^\infty \left((t_1-T+s)^{H-\frac{1}{2}} 
 - (t_2-T +s)^{H-\frac{1}{2}}\right)^2 ds\right)^{1/2}\\
&\leq \left(\int_0^\infty \left(s^{H-\frac{1}{2}} -
  (t_2-T +s)^{H-\frac{1}{2}}\right)^2 ds\right)^{1/2}
  \leq
  C(t_2-T)^H\leq 2C (t_2-t_1) (t_2-T)^{H-1}. 
\end{align*}

Consider now the trapezoidal scheme. 
On the one hand, when  $t_3-T \geq 2(t_2-t_1)$, a straightforward second-order Taylor estimate plus the
same argument as above yield the following bound:
\begin{align*}
&\left(\int_t^T 
\left\|g(t_2,s)-\frac{t_3-t_2}{t_3-t_1}g(t_1,s)-  \frac{t_2-t_1}{t_3-t_1}g(t_3,s)\right\|^2ds\right)^{1/2}\\
& = \left(\int_t^T \left|(t_2-s)^{H-\frac{1}{2}}-\frac{t_3-t_2}{t_3-t_1}(t_1-s)^{H-\frac{1}{2}}-
  \frac{t_2-t_1}{t_3-t_1}(t_3-s)^{H-\frac{1}{2}}\right|^2ds\right)^{1/2}\\
  &\leq
\left(\int_t^{T} \left|((t_1+t_3)/2-s)^{H-\frac{1}{2}} - \frac{1}{2}(t_1 - s)^{H-\frac{1}{2}} -
  \frac{1}{2}(t_3 -s)^{H-\frac{1}{2}}\right|^2 ds\right)^{1/2} \\
&\leq C (t_3-t_1)^2 (t_1-T)^{H-2} \leq C 2^{2-H} (t_3-t_1)^2 (t_3-T)^{H-2}.
\end{align*}
When $t_2-T < 2(t_2-t_1)$, this expression may be bounded as follows, completing the
proof:
$$
\left(\int_t^{T} \left|\left(\frac{t_1+t_3}{2}-s\right)^{H-\frac{1}{2}} - \frac{(t_1 - s)^{H-\frac{1}{2}}}{2}
 - \frac{(t_3 -s)^{H-\frac{1}{2}}}{2}\right|^2 ds\right)^{\frac{1}{2}}
\leq C (t_3-T)^H\leq 4C (t_3-t_1)^2 (t_3-T)^{H-2}. 
$$

\subsection{A key lemma}
\begin{lemma}\label{density.lm}
Let $(Z_t)_{0\leq t\leq 1}$ be an a.s.~continuous Gaussian process 
with mean function~$m(\cdot)$ and continuous covariance function $C(\cdot,\cdot)$, satisfying
$$
\min_{u\in [0,1]} \int_0^1 C(u,v)dv>0. 
$$ 
Let $x\in L^2([0,1])$ with $x\geq 0$. 
Then the random variable
$\int_0^1  x(u) e^{Z(u)}du$
has a density $p(\cdot)$ on $(0,\infty)$, such that $p(x)\leq c/x$ for some finite constant~$c$. 
\end{lemma}
\begin{proof}
{Let $\phi$ be a smooth bounded function with bounded derivative, denote
$t^n_i := \frac{i}{n}$ and $\zeta^n_i := \int_{t_i^n}^{t_{i+1}^n} x(u)
du$, and for $i=0,\dots,n-1$, let $Z^\varepsilon_i = Z_{t^n_i} +
\varepsilon \zeta_i$, where $(\zeta_0,\dots,\zeta_{n-1})$ is a
standard normal random vector independent from $Z$. Then,
$$
\EE\left[\phi\left(\sum_{i=0}^{n-1} \zeta^n_i e^{Z^\varepsilon_i}\right)\right]
 = \int_{\mathbb{R}^n}\phi\left(\sum_{i=0}^{n-1} \zeta^n_i e^{z_i}\right)
 \frac{\exp\left(-\frac{1}{2} (z-m_n)^\top (C^\varepsilon_n)^{-1} (z-m_n)\right)}{(2\pi)^{n/2}\sqrt{|C^\varepsilon_n|}}
 dz,
$$
with $z=(z_1,\ldots,z_n)$,
where $C^\varepsilon_n$ and $m_n$ are, respectively, the covariance matrix and the
mean vector of the Gaussian vector
$(Z^\varepsilon_{0},\dots,Z^\varepsilon_{n-1})$ ($C^\varepsilon_n$ is
clearly nondegenerate for any $\varepsilon>0$).
Making the change of variable $z_i\mapsto z_i + \alpha \rho_i$ in the
above integral, differentiating with respect to $\alpha$ and taking
$\alpha=0$, we obtain
\begin{align*}
&\int_{\mathbb{R}^n}\phi'\left(\sum_{i=0}^{n-1} \zeta^n_i  e^{z_i}\right) \sum_{i=0}^{n-1} \zeta^n_i \rho_i e^{z_i}
  \frac{\exp\left(-\frac{1}{2} (z-m_n)^\top (C^\varepsilon_n)^{-1} (z-m_n)\right)}{(2\pi)^{n/2}\sqrt{|C^\varepsilon_n|}}dz \\
& = \int_{\mathbb{R}^n}\phi\left(\sum_{i=0}^{n-1} \zeta^n_i
  e^{z_i}\right)\frac{\rho^\top (C^\varepsilon_n)^{-1} (z-m_n)}{(2\pi)^{n/2}
  \sqrt{|C^\varepsilon_n|}}\exp\left(-\frac{1}{2} (z-m_n)^\top (C^\varepsilon_n)^{-1} (z-m_n)\right)dz,
\end{align*}
or, in other words,
$$
\EE\left[\phi'(X^\varepsilon_n) \sum_{i=0}^{n-1} \zeta^n_i \rho_i e^{Z^\varepsilon_{i}}\right] = \EE\left[\phi\left(X^\varepsilon_n
  \right) \rho^\top (C^\varepsilon_n)^{-1}  (Z^\varepsilon-m_n)\right]
  $$
where we define $X^\varepsilon_n := \sum_{i=0}^{n-1} \zeta^n_i
e^{Z^\varepsilon_{i}}$.
Taking $\rho = \mathbf 1^\top C^\varepsilon_n$, we then have
$$
 \EE\left[\phi'(X^\varepsilon_n) \sum_{i=0}^{n-1} \zeta^n_i [\mathbf 1^\top C^\varepsilon_n]_i e^{Z^\varepsilon_{i}}\right]= \EE\left[\phi(X^\varepsilon_n) \mathbf 1^\top (Z^\varepsilon-m_n)\right],
$$
and applying the dominated convergence theorem, we finally get
$$
 \EE\left[\phi'(X_n) \sum_{i=0}^{n-1} \zeta^n_i [\mathbf 1^\top C_n]_i e^{Z_{t^n_i}}\right]= \EE\left[\phi(X_n) \mathbf 1^\top (Z_{t^n_\cdot}-m_n)\right],
$$
where $C_n$ is the covariance matrix of $Z_{t^n_0},\dots,
Z_{t^n_{n-1}}$, even if this matrix is degenerate. }

Assuming that $|\phi|\leq 1$, Jensen's inequality implies
$$
\EE\left[\phi'(X_n) \sum_{i=0}^{n-1} \zeta^n_i [\mathbf 1^\top C_n]_i Z_{t^n_i}\right]
 \leq \EE\left[(\mathbf 1^\top (Z_{t^n_\cdot}-m_n))^2\right]^{\frac{1}{2}}
  =(\mathbf 1^\top C_n\mathbf 1)^{\frac{1}{2}}. 
$$
Now, multiplying both sides by $\frac{1}{n}$ and passing to the limit as~$n$ tends to infinity yield
$$
\EE\left[\phi'(X) \int_0^{1}x(u)e^{Z_u}\int_{0}^{1} C(u,v) dvdu\right] \leq \int_{0}^{1} du \int_0^{1} dv C(u,v),
$$
so that
$$
\EE[X\phi'(X)] \leq \frac{\int_{0}^{1} du \int_0^{1} dv C(u,v)}{\min_u\int_{0}^{1} C(u,v) dv}. 
$$
The passage to the limit is justified by the continuity of the
covariance function and that of~$Z$, together with the dominated
convergence theorem and standard bounds on Gaussian processes. Now,
for $0<a<b<\infty$, 
choose $\phi(x) = 0$ for $x<a$, $\phi(x) =
\frac{\log(x/a)}{\log(b/a)}$ for $a\leq x<b$ and $\phi(x) = 1$ for
$b\leq x$. The above bound becomes
$$
\frac{\mathbb P[a<X<b]}{\log(b/a)} \leq \frac{\int_{0}^{1} du \int_0^{1} dv
  C(u,v}{\min_u\int_{0}^{1} C(u,v) dv},
$$
from which the statement of the lemma follows directly. 
\end{proof}

\subsection{Proof of Proposition~\ref{mc.prop}}\label{app:Proofmc.prop}
In the proof, $C$ denotes a constant, not depending on~$n$, which may change from line to line. 
By the Lipschitz property of~$f$, using the positivity of the path~$x(\cdot)$, 
\begin{align}
|F(t,x)-F_n(t,x)| & \leq  
\frac{C}{\Theta} \int_{T}^{T+\Theta} x(u) \EE\left[\left|\Ee_{t, {T}}(u) - \Ee_{t, {T}}(\eta(u))\right|\right] du,\label{rectapprox}\\
|F(t,x)-\widehat F_n(t,x)| &\leq  \frac{C}{\Theta}\int_{T}^{T+\Theta} x(u) 
\EE\left[\left|\Ee_{t, {T}}(u) - \theta^n (u)\Ee_{t,T}\left(\eta^n(u)\right) - \left(1-\theta^n(u)\right) \Ee_{t,T}(\overline{\eta}^n(u))\right|\right] du.\label{trapapprox}
\end{align}
We now estimate the expectation under the integral sign for the two approximations. 
For the rectangle approximation, for $u\in [t^n_i,t^n_{i+1})$, 
there exists $\theta\in [0,1]$ (possibly random) such that 
\begin{align*}
&\EE\left(\left|\Ee_{t, {T}}(u) - \Ee_{t, {T}}(t^n_i)\right|\right)
 = \EE\left[|Z_u - Z_{t^n_i}| e^{Z_{t^n_i} + \theta\left(Z_u - Z_{t^n_i}\right)}\right]\\ 
& = 2 \EE\left[\left|\int_t^{T} (g(u,s)-g(t^n_i,s))^\top dW_s\right|e^{Z_{t^n_i} + \theta(Z_u - Z_{t^n_i})}\right]  + 2 \left|\int_t^{T} (\|g(u,s)\|^2 - \|g(t^n_i,s)\|^2) ds\right| 
 \EE\left[e^{Z_{t^n_i} + \theta(Z_u - Z_{t^n_i})}\right]\\
& \leq 2 \EE\left[\int_t^{T} \|g(u,s)-g(t^n_i,s)\|^2 ds\right]^{1/2} 
\EE\left[e^{2Z_{t^n_i} + 2\theta(Z_u -  Z_{t^n_i})}\right]^{1/2}  + 2 \left|\int_t^{T} (\|g(u,s)\|^2 - \|g(t^n_i,s)\|^2) ds\right|
 \EE\left[e^{Z_{t^n_i} + \theta(Z_u - Z_{t^n_i})}\right]\\
&\leq  2 \EE\left[\int_t^{T} \|g(u,s)-g(t^n_i,s)\|^2 ds\right]^{1/2} \EE\left(e^{2 M}\right)^{1/2}
+ 2 \left|\int_t^{T} (\|g(u,s)\|^2 - \|g(t^n_i,s)\|^2) ds\right| \EE\left(e^{M}\right),
\end{align*}
{with~$Z$ defined in~\eqref{eq:ZProcess}, }
where $M := \max_{{T}\leq u\leq T+\Theta } Z_u$. 
By classical results on the supremum of Gaussian processes~\cite[Section 2.1]{adler2007random}, 
$M$ admits all exponential moments, and therefore 
$$
\EE[|\Ee_{t,{T}}(u) - \Ee_{t,{T}}(t^n_i)|]
  \leq C \EE\left[\int_t^{T} \|g(u,s)-g(t^n_i,s)\|^2 ds\right]^{1/2}
   + C\left|  \int_t^{T} (\|g(u,s)\|^2 - \|g(t^n_i,s)\|^2) ds\right|.
$$
The second term above can be further estimated as 
\begin{align*}
\left|
  \int_t^{T} (\|g(u,s)\|^2 - \|g(t^n_i,s)\|^2) ds\right|  & \leq
\int_t^{T} \|g(u,s)-g(t^n_i,s)\|^2 ds + 2 \int_t^T\|g(u,s)-g(t^n_i,s)\|\|g(t^n_i,s)\|ds\\
&\leq c^2 (u-t^n_i)^2 (u-T)^{2\beta-2}
 + 2C \left[\int_t^{T}  \|g(u,s)-g(t^n_i,s)\|^2 ds\right]^{\frac{1}{2}}
 \left[\int_t^{T}\|g(t^n_i,s)\|^2 ds\right]^{\frac{1}{2}} \\
&\leq C (t^n_{i+1}-t^n_i) (u-T)^{\beta-1},
\end{align*}
where in the last estimate we have used Assumption~\eqref{intg.eq}. Finally, 
$$
\EE[|\Ee_{t,{T}}(u) - \Ee_{t,{T}}(t_i)|]   \leq C (t^n_{i+1}-t^n_i) (u-T)^{\beta-1},
$$
and the proposition follows from the integrability of $(u-T)^{\beta-1}$ on $[T,T+\Theta]$ and the boundedness of $x$. 
For the trapezoidal approximation, for $u\in[t^n_i,t^n_{i+1})$, 
\begin{align*}
& \EE[|\Ee_{t, {T}}(u) - \theta^n (u)\Ee_{t,T}(t^n_i) - (1-\theta^n(u)) \Ee_{t,T}(t^n_{i+1})|]  \\
& \leq \EE\left[\left|e^{Z_u} - e^{\theta^n (u)Z_{t^n_i} + (1-\theta^n(u))
  Z_{t^n_{i+1}}}\right|\right]
   +\EE\left[\left| e^{\theta^n (u)Z_{t^n_i} + (1-\theta^n(u))
  Z_{t^n_{i+1}}}- \theta^n (u)e^{Z_{t^n_i}} - (1-\theta^n(u)) e^{Z_{t^n_{i+1}}}\right|\right] \\
& \leq \EE\left[\left|Z_u - \theta^n (u)Z_{t^n_i} - (1-\theta^n(u))  Z_{t^n_{i+1}}\right| 
e^{\theta_1 Z_u + \theta_2 Z_{t^n_i} + \theta_3  Z_{t^n_{i+1}}}\right]
 + C\EE\left[e^{Z_{t^n_i} + |Z_{t^n_{i+1}}-Z_{t^n_i}|} \left|Z_{t^n_{i+1}}-Z_{t^n_i}\right|^2\right]
\end{align*}
for some (possibly random) $\theta_1,\theta_2,\theta_3\geq 0$ with
$\theta_1+\theta_2+\theta_3 = 1$. Cauchy-Schwarz inequality
and the exponential integrability of the supremum of Gaussian processes then yield
\begin{align*}
&\EE[|\Ee_{t, {T}}(u) - \theta^n (u)\Ee_{t,T}(t^n_i) - (1-\theta^n(u)) \Ee_{t,T}(t^n_{i+1})|]  \\
&\leq C E\left[\left|Z_u - \theta^n (u)Z_{t^n_i} - (1-\theta^n(u))  Z_{t^n_{i+1}}\right|^2\right]^{1/2}
 + C \EE\left[\left|Z_{t_{i+1}}-Z_{t_i}\right|^4\right]^{1/2}\\
& \leq C \left(\int_t^{T} \|g(u,s) - \theta^n (u)g(t^n_i, s) -
  (1-\theta^n(u))g(t^n_{i+1},s) \|^2 ds\right)^{1/2} \\
   &+ C\Big|\int_t^{T} (\|g(u,s)\|^2 - \theta^n (u)\|g(t^n_i,s)\|^2 -
  (1-\theta^n(u))\|g(t^n_{i+1},s)\|^2 )ds\Big|\\
 & +  C \int_t^{T} \|g(t^n_i,s) -  g(t^n_{i+1},s) \|^2 ds + C \Big|\int_t^{T} (\|g(t^n_i,s)\|^2 -
   \|g(t^n_{i+1},s)\|^2  )ds\Big|^{2}\\
 & \leq C(t^{n}_{i+1} -t^n_i)^2 (t^{n}_{i+1}-T)^{\beta-2} +
   C(t^{n}_{i+1} -t^n_i)^2 (t^{n}_{i+1}-T)^{2\beta-2}\\
&+ C  \Big|\int_t^{T} (\|g(u,s)\|^2 - \theta^n (u)\|g(t^n_i,s)\|^2 - (1-\theta^n(u))\|g(t^n_{i+1},s)\|^2 )ds\Big|
\end{align*}
where, for a centered Gaussian random variable $X$, $\EE[X^4] = 3 \EE[X^2]^2$, 
and we used the estimate of the first part of the proof. The remaining term is estimated as 
\begin{align*}
&\Big|\|g(u,s)\|^2 - \theta^n (u)\|g(t^n_i,s)\|^2 -
  (1-\theta^n(u))\|g(t^n_{i+1},s)\|^2 \Big|\\
&\leq \|g(u,s) -\theta^n (u) g(t^n_i,s) -(1-\theta^n(u))g(t^n_{i+1},s) \|^2\\
&+ 2 \|g(u,s) -\theta^n (u) g(t^n_i,s) -(1-\theta^n(u))g(t^n_{i+1},s) \| \|\theta^n (u) g(t^n_i,s) +(1-\theta^n(u))g(t^n_{i+1},s)\|\\
&+ \Big|\|\theta^n (u) g(t^n_i,s) +(1-\theta^n(u))g(t^n_{i+1},s)\|^2 - \theta^n (u)\|g(t^n_i,s)\|^2 -
  (1-\theta^n(u))\|g(t^n_{i+1},s)\|^2 \Big|.
\end{align*}
Each of the three terms can now be estimated similarly to the first
part of the proof, leading to the conclusion
$$
\EE\left[\left|\Ee_{t, {T}}(u) - \theta^n (u)\Ee_{t,T}(t^n_i) - (1-\theta^n(u)) \Ee_{t,T}(t^n_{i+1})\right|\right]  
\leq C(t^{n}_{i+1} -t^n_i)^2 (t^{n}_{i+1}-T)^{\beta-2}. 
$$
Using the boundedness of $x(u)$, estimating the discretisation error now boils down to computing
$$
\sum_{i=1}^n (t^n_{i+1}-t^n_i)^3(t^n_{i+1} - T)^{H-2}.
$$
Letting ${\Theta} = 1$ without loss of generality, and substituting the
expression for $t^n_i$, this becomes
$$
\sum_{i=1}^n
\left[\left(\frac{i+1}{n}\right)^\kappa-\left(\frac{i}{n}\right)^\kappa\right]^3\left(\frac{i+1}{n}\right)^{\kappa(H-2)}
\leq  C\left\{n^{-\kappa(H+1)}\sum_{i=1}^n (i+1)^{\kappa(H+1)-3} + n^{-\kappa(H+1)} \right\}.
$$
When $\kappa(H+1)>2$, the sum in the right-hand side explodes at the
rate $n^{\kappa(H+1)-2}$, and therefore the entire right-hand side is
bounded by $C/n^{2}$. 

\subsection{Proof of Proposition~\ref{mc2.prop}}
In the proof, $C$ denotes a constant, not depending on $n$, which may
change from line to line.  We drop the superscript $(t,\gamma)$ whenever this does not cause confusion. Similarly to the proof of Proposition~\ref{mc.prop}, we need to estimate the expectation under the integral
sign for~\eqref{rectapprox} and~\eqref{trapapprox}. 
We denote
$$
\Ee_{t,T}(u) = \exp\left(Z_{t,T}(u)\right):= \exp\left(Z^1_{t,T}(u) + Z^2_{t,T}(u)\right),
$$
with 
$$
Z^1_{t,T}(u):= 2\int_t^T  \sqrt{\Gamma_s} g(u-s)^\top  dW_s
$$
and 
$$
Z^2_{t,T}(u)= \psi(u-T)\Gamma_T +
  \phi(u-T)-\psi(u-t)\gamma-\phi(u-t).
$$

\paragraph{Part i.} By Taylor formula,
\begin{align*}
&|\Ee_{t, {T}}(u) - \Ee_{t, {T}}(t^n_i)| \leq |Z_{t,T}(u)
- Z_{t,T}(t^n_i)|e^{Z_{t,T}(u)} +|Z_{t,T}(u)
- Z_{t,T}(t^n_i)|e^{Z_{t,T}(t^n_i)}.
\end{align*}
Let us focus, for example, on the first term; the second one can be dealt with in a similar manner. First, by taking a conditional expectation and using the positivity of $\psi$ and $\phi$, 
$$
\EE[|\Ee_{t, {T}}(u) - \Ee_{t, {T}}(t^n_i)|] \leq
\EE\left[|Z_{t,T}(u) - Z_{t,T}(t^n_i)|
\exp\left(Z^1_{t,T}(u) + \int_T^u \Gamma_s \|g(u-s)\|^2 ds\right)\right]
$$ 
Using the triangle inequality, the Cauchy-Schwarz inequality, and the
It\^o isometry we
then get
\begin{align}
\EE\left[|\Ee_{t, {T}}(u) - \Ee_{t, {T}}(t^n_i)|\right]
 & \leq \left(\EE\left[|Z^1_{t,T}(u) - Z^1_{t,T}(t^n_i)|^2\right]^{1/2}
  +\EE\left[|Z^2_{t,T}(u) - Z^2_{t,T}(t^n_i)|^2\right]^{1/2} \right)\notag\\ 
&\times \EE\left[\exp\left(4\int_t^T \Gamma_s \|g(u-s)\|^2 ds+ 2\int_T^u \Gamma_s \|g(u-s)\|^2 ds\right)\right]^{1/2}
\notag\\
&\leq C \EE[|Z^1_{t,T}(u)
- Z^1_{t,T}(t^n_i)|^2]^{1/2} + C\EE[|Z^2_{t,T}(u)
- Z^2_{t,T}(t^n_i)|^2]^{1/2},  
\label{2factors}
\end{align}
because the last factor is bounded uniformly on $u$ by assumptions of
this proposition. 

The first summand above satisfies
\begin{align*}
&\EE[|Z^1_{t,T}(u)
- Z^1_{t,T}(t^n_i)|^2]^{\frac{1}{2}}\leq 2\EE\left[\int_t^T
  \Gamma_s \|g(u-s)-g(t^n_i-s)\|^2 ds\right]^{\frac{1}{2}} \\&\leq 2\sqrt{\EE[\max_{t\leq s\leq T}\Gamma_s]} \sqrt{\int_t^T
  \|g(u-s)-g(t^n_i-s)\|^2 ds} \leq C\sqrt{\int_t^T
  \|g(u-s)-g(t^n_i-s)\|^2 ds}.
\end{align*}
The contribution of this term to the global error is of order $\frac{1}{n}$ 
as in Proposition~\ref{mc.prop}. 
To estimate the second summand, remark that by Proposition~\ref{affine.prop}, it follows that
$$
|\psi(u-t) - \psi(t^n_i-t)|\leq C |u-t^n_i| + 2\int_{t^n_i-t}^{u-t}
\|g(s)\|^2 ds\quad \text{and}\quad |\phi(u-t) - \phi(t^n_i-t)|\leq C |u-t^n_i|,
$$
for some constant $C<\infty$, and therefore,
$$
\EE\left[\left|Z^2_{t,T}(u) - Z^2_{t,T}(t^n_i)\right|^2\right]^{\frac{1}{2}}
\leq C\left(|u-t^n_i| +\int_{t^n_i-t}^{u-t}
\|g(s)\|^2 ds +\int_{t^n_i-T}^{u-T}
\|g(s)\|^2 ds\right)
$$
Remark also that 
$$
\int_{t^n_i}^{t^n_{i+1}} du \int_{t^n_i-T}^{u-T}
\|g(s)\|^2 ds = \int_{t^n_i}^{t^n_{i+1}} ds (t^n_{i+1}-u) \|g(s-T)\|^2
ds\leq \frac{\Theta}{n} \int_{t^n_i}^{t^n_{i+1}} ds\|g(s-T)\|^2
ds.
$$
Therefore, the statement of the proposition for the rectangle scheme
follows from the integrability of $\|g(s)\|^2$.

\paragraph{Part ii.} For the trapezoidal approximation, by Taylor formula, 
for $u\in[t^n_i,t^n_{i+1})$, 
\begin{align*}
&|\Ee_{t, {T}}(u) -     \theta^n (u)\Ee_{t,T}(t^n_i) - (1-\theta^n(u)) \Ee_{t,T}(t^n_{i+1})| \\
& \leq \left|e^{Z_u} - e^{\theta^n (u)Z_{t^n_i} + (1-\theta^n(u))Z_{t^n_{i+1}}} \right|
+ \left| e^{\theta^n (u)Z_{t^n_i} + (1-\theta^n(u)) Z_{t^n_{i+1}}}- 
    \theta^n (u)e^{Z_{t^n_i}} - (1-\theta^n(u)) e^{Z_{t^n_{i+1}}}\right|\\
& \leq\left(\left|Z_u - \theta^n (u)Z_{t^n_i} - (1-\theta^n(u))
  Z_{t^n_{i+1}}\right|  + \left|Z_{t_{i+1}}-Z_{t_i}\right|^2\right)
  \left(e^{Z_u} + e^{Z_{t^n_i}}+ e^{Z_{t^n_{i+1}}}\right).
\end{align*}
Similarly to the first part of the proof, we can then show using the
Cauchy-Schwarz inequality that 
\begin{align}
&\EE\left[\left|\Ee_{t, {T}}(u) - 
    \theta^n (u)\Ee_{t,T}(t^n_i) - (1-\theta^n(u)) \Ee_{t,T}(t^n_{i+1})\right|\right] \notag\\&
  \leq C \EE\left[\left|Z^1_u - \theta^n (u)Z^1_{t^n_i} - (1-\theta^n(u))
  Z^1_{t^n_{i+1}}\right|^2\right]^{1/2} + C \EE\left(|Z^1_{t_{i+1}}-Z^1_{t_i}|^4\right)^{1/2}
  \label{1stline}\\
& + C\EE\left[\left|Z^2_u - \theta^n (u)Z^2_{t^n_i} - (1-\theta^n(u))
  Z^2_{t^n_{i+1}}\right|^2\right]^{1/2}
   + C \EE\left(\left|Z^2_{t_{i+1}}-Z^2_{t_i}\right|^4\right)^{1/2}\label{2ndline}
\end{align}
The two terms in~\eqref{1stline} are estimated using It\^o isometry:
\begin{align*}
\EE\left[\left|Z^1_u - \theta^n (u)Z^1_{t^n_i} - (1-\theta^n(u))
  Z^1_{t^n_{i+1}}\right|^2\right]^{\frac{1}{2}} 
  & \leq \EE\left[\int_t^T\Gamma_s \left\|
  g(u-s) - \theta^n(u) g(t^n_i - s) - (1-\theta^n_u) g(t^n_{i+1}-s)\right\|^2 ds\right]^{\frac{1}{2}}\\
&\leq\sqrt{\max_{t\leq s\leq T} \EE[\Gamma_s]}
  (t^n_{i+1}-t^n_i)^2 (t^{n}_{i+1}-T)^{\beta-2},
\end{align*}
and similarly
\begin{align*}
\EE\left(|Z^1_{t_{i+1}}-Z^1_{t_i}|^4\right)^{1/2}
 & \leq 3\sqrt{\max_{t\leq s\leq T} \EE[\Gamma^2_s]}  \EE\left[\left(\int_t^T \|g(t^n_{i+1}-s) - g(t^n_i-s)\|^2 ds\right)^2\right]^{1/2}\\
& \leq C\sqrt{\max_{t\leq s\leq T} \EE[\Gamma^2_s]}   (t^n_{i+1}-t^n_i)^2(t^{n}_{i+1}-T)^{2\beta-2}.
\end{align*}
The contribution of these terms to the final error estimate is
therefore the same as in Proposition~\ref{mc.prop}. 

It remains to estimate the contribution of the terms in~\eqref{2ndline}. 
From Proposition~\ref{affine.prop}, both $\phi''$ and~$\bar{\phi}_0''$ 
(introduced in the proof of Proposition~\ref{affine.prop}) are bounded on $[0,T]$. 
Therefore
\begin{align*}
\left|Z^2_u - \theta^n (u)Z^2_{t^n_i} - (1-\theta^n(u))  Z^2_{t^n_{i+1}}\right|
   & \leq C(1+\gamma + \Gamma_T ) (t^n_{i+1}-t^n_i)^2 \\
   &+ \Gamma_T|G(u-T) - \theta^n(u) G(t^n_i-T) -  (1-\theta^n(u))G(t^n_{i+1}-T)|\\
   & + \gamma|G(u-t) - \theta^n(u) G(t^n_i-t) - (1-\theta^n(u))G(t^n_{i+1}-t)|.
\end{align*}
Since $\|g\|$ is decreasing, $G$ is concave and for $i\geq 1$
\begin{align*}
&\int_{t^n_i}^{t^n_{i+1}}|G(u-T) - \theta^n(u) G(t^n_i-T) -
  (1-\theta^n(u))G(t^n_{i+1}-T)|du\\
& = \int_{t^n_i}^{t^n_{i+1}}\left(\int_{t^n_i }^{u} \|g(s-T)\|^2 ds
  -  (1-\theta^n(u)) \int_{t^n_i }^{t^n_{i+1}} \|g(s-T)\|^2
  ds\right)du\\
&= \int_{t^n_i}^{t^n_{i+1}} 
  \|g(s-T)\|^2 \left(\frac{t^n_{i+1}+t^n_i}{2}-s\right)ds\\
& \leq  \|g(t^n_i-T)\|^2\int_{t^n_i}^{\frac{t^n_{i+1}+t^n_i}{2}} 
  \left(\frac{t^n_{i+1}+t^n_i}{2}-s\right)ds
  - \|g(t^n_{i+1}-T)\|^2\int_{\frac{t^n_{i+1}+t^n_i}{2}}^{t^n_{i+1}} 
  \left(s-\frac{t^n_{i+1}+t^n_i}{2}\right)ds\\
 & = \frac{(t^n_{i+1}-t^n_i)^2}{8}
  \left(\|g(t^n_i-T)\|^2 -\|g(t^n_{i+1}-T)\|^2\right)
 \leq C (t^n_{i+1}-t^n_i)^3
  (t^n_i-T)^{\beta-2} \leq C' (t^n_{i+1}-t^n_i)^3
  (t^n_{i+1}-T)^{\beta-2},
\end{align*}
where $C'$ is a different constant. On the other hand, for $i=0$, the
same inequality is obtained from the bound on $\|g\|$. Finally,
\begin{align*}
|Z^2_{t_{i+1}}-Z^2_{t_i}|&\leq C(1+\gamma+\Gamma_T)(t^n_{i+1}-t^n_i)
+\Gamma_T|G(t^n_i-T)-G(t^n_{i+1}-T)|
 + \gamma|G(t^n_i-t)-G(t^n_{i+1}-t)|\\
& = C(1+\gamma+\Gamma_T)(t^n_{i+1}-t^n_i)
+ \Gamma_T\int_{t^n_i}^{t^n_{i+1}} \|g(s-T)\|^2 ds
+\gamma\int_{t^n_i}^{t^n_{i+1}} \|g(s-t)\|^2  ds. 
\end{align*}
Using the bound on $g$ and the integrability of $\Gamma_T^4$, 
treating once again separately the case $i=0$, we find
$$
\EE\left(|Z^2_{t_{i+1}} - Z^2_{t_i}|^4\right)^{\frac{1}{2}} \leq C (t^n_{i+1}-t^n_i)^2(t^n_{i+1}-T)^{2\beta-2},
$$
so that the terms in~\eqref{2ndline} have the same contribution to the
error as the terms in~\eqref{1stline}, and the proof follows. 

\begin{thebibliography}{99}

\bibitem{Larsson} 
{\sc E. Abi Jaber, M. Larsson and S. Pulido.}
{\em Affine Volterra processes.}
\href{https://arxiv.org/abs/1708.08796}{arXiv:1708.08796}, 2017.

\bibitem{adler2007random}
{\sc R.~J. Adler and J.~E. Taylor}.
 {\em Random fields and geometry}. 
Springer,  2007.

\bibitem{Akdogan}
{\sc O. Akdogan}
{\em Variance curve models: Finite dimensional realisations and beyond}.
PhD Thesis, ETH Z\"urich, \href{https://www.research-collection.ethz.ch/bitstream/handle/20.500.11850/156070/eth-50162-02.pdf?sequence=2&isAllowed=y}{research-collection.ethz.ch} 2016.

\bibitem{ALV07}
{\sc E. Al\`os, J. Le\'on and J. Vives}. 
{\em On the short-time behavior of the implied volatility for jump-diffusion models with stochastic volatility}. 
Finance and Stochastics, {\tt 11}(4), 571-589, 2007.

\bibitem{bargourleipold}
{\sc C. Bardgett, E. Gourier and M. Leippold}.
{\em Inferring volatility dynamics and risk premia from the S\&P 500 and VIX markets}. 
Forthcoming in Journal of Financial Economics.

\bibitem{BSS}
{\sc O.E. Barndorff-Nielsen and J. Schmiegel}. 
{\em Brownian semistationary processes and volatility/intermittency}.
In H. Albrecher, W. Runggaldier, and W. Schachermayer, editors, Advanced Financial Modelling: 1-26. 
Walter de Gruyter, Berlin, 2009.

\bibitem{bayer2016pricing}
{\sc C.~Bayer, P.~Friz, and J.~Gatheral}. 
{\em Pricing under rough volatility}.
Quantitative Finance, {\tt 16}: 887-904, 2016.

\bibitem{BayerRoughSkew}
{\sc C.~Bayer, P. K. Friz, A.~Gulisashvili, B.~Horvath and B.~Stemper}.
{\em Short-time near-the-money skew in rough fractional volatility models}.
Forthcoming in Quantitative Finance, 2019.

\bibitem{BayerDMR}
{\sc C. ~Bayer, J. Gatheral, M. Karlsmark}
{\em Fast Ninomiya-Victoir calibration of the Double-Mean-Reverting Model}.
Quantitative Finance, {\tt 13}(11): 1813-1829, 2013. 

\bibitem{bennedsen2017hybrid}
{\sc M.~Bennedsen, A.~Lunde, and M.~S. Pakkanen}.
{\em Hybrid scheme for Brownian semistationary processes}.
Finance and Stochastics, {\tt 21}: 931-965, 2017.

\bibitem{BergomiSmileDynamicsII}
{\sc L. Bergomi}.
{\em Smile Dynamics II}.
Risk, {\tt October}: 67-73, 2005.

\bibitem{BergomiSmileDynamicsIII}
{\sc L. Bergomi}.
{\em Smile Dynamics III}.
Risk, {\tt October}: 90-96, 2008.

\bibitem{Buhler}
{\sc H. B\"uhler.} 
{\em Volatility markets: Consistent modeling, hedging, and practical implementation of variance swap market models.} PhD Thesis, Technischen Universit\"at Berlin, 2006.

\bibitem{carrmadan1}
{\sc P. Carr and D. Madan}.
{\em Towards a theory of volatility trading}, Risk: 417-427, 1998.

\bibitem{carrmadan2}
{\sc P. Carr and D. Madan}.
{\em Joint modeling of VIX and SPX options at a single and common maturity with risk management
applications}.
IIE Transactions, {\tt 46}(11): 1125-1131, 2014.

\bibitem{daprato2014stochastic}
{\sc G.~Da~Prato and J.~Zabczyk}.
{\em Stochastic equations in infinite dimensions, 2nd Edition}.
Cambridge University Press, 2014.

\bibitem{duffie2003affine}
{\sc D.~Duffie, D.~Filipovi{\'c} and W.~Schachermayer}.
{\em Affine processes and applications in Finance}.
Annals of Applied Probability, {\tt 13}(3): 984-1053, 2003.

\bibitem{omar2016microstructural}
{\sc O.~El~Euch, M.~Fukasawa and M.~Rosenbaum}.
{\em The microstructural  foundations of leverage effect and rough volatility}.
Finance and Stochastics, {\tt 22}(2): 241-280, 2018.

\bibitem{euch2016characteristic}
{\sc O.~El Euch and M.~Rosenbaum}.
{\em The characteristic function of rough Heston models}.
Mathematical Finance, {\tt 29}(1): 3-38, 2019.

\bibitem{Euch1}
{\sc O. El Euch and M. Rosenbaum}.
{\em Perfect hedging in rough Heston models}.
The Annals of Applied Probability, {\tt 28}(6): 3813-3856, 2018.

\bibitem{Fukasawa}
{\sc M. Fukasawa}.
{\em Short-time at-the-money skew and rough fractional volatility}.
Quantitative Finance, {\tt 17}(2): 189-198, 2017.

\bibitem{Fukasawatalk}
{\sc M. Fukasawa}.
{\em Hedging and Calibration for Log-normal Rough Volatility Models}.
Presentation, Jim Gatheral's 60th Birthday Conference, New York, 2017.

\bibitem{gatheralslides}
{\sc J.~Gatheral}.
{\em Consistent modelling of SPX and VIX options}.
Presentation, Bachelier Congress, London, 2008.

\bibitem{gatheral2014volatility}
{\sc J.~Gatheral, T.~Jaisson, and M.~Rosenbaum}.
{\em Volatility is rough}.
Quantitative Finance, {\tt 18}(6): 933-949, 2018.

\bibitem{GKR18}
{\sc J. Gatheral and M. Keller-Ressel}.
{\em Affine forward variance models}.
\href{https://papers.ssrn.com/sol3/papers.cfm?abstract_id=3105387}{SSRN:3105387}, 2018.

\bibitem{grad}
{\sc I.~S. Gradshteyn and I.~M. Ryzhik}.
{\em Table of integrals, series, and products}.
Academic press, 2014.

\bibitem{GJRSHeston}
{\sc H. Guennoun, A. Jacquier P. Roome and F. Shi}
{\em Asymptotic behaviour of the fractional Heston model}. 
\textit{SIAM Journal on Financial Mathematics}, {\tt 9}(3), 1017-1045, 2018.

\bibitem{GuyonNutz}
{\sc J. Guyon, R. Menegaux and M. Nutz}.
{\em Bounds for VIX Futures given S\&P 500 Smiles } 
Finance \& Stochastics, {\tt 21}(3): 593-630, 2017.

\bibitem{martiniessvi}
{\sc S. Hendriks and C. Martini}.
{\em The extended SSVI volatility surface}.
\href{https://papers.ssrn.com/sol3/papers.cfm?abstract_id=2971502}{SSRN:2971502}, 2017.

\bibitem{hida93gaussian}
{\sc T.~Hida and M.~Hitsuda}
{\em Gaussian Processes}. American Mathematical Society, 1993. 

\bibitem{HJL2017}
{\sc B. Horvath, A. Jacquier and C. Lacombe}.
{\em Asymptotic behaviour of randomised fractional volatility models}.
\href{https://arxiv.org/abs/1708.01121}{arXiv:1708.01121}, 2017.

\bibitem{HJMTrees}
{\sc B. Horvath, A. Jacquier and A. Muguruza}.
{\em Functional central limit theorems for rough volatility}.
\href{https://arxiv.org/abs/1711.03078}{arXiv:1711.03078}, 2017.

\bibitem{jacod2013limit}
{\sc J.~Jacod and A.~Shiryaev}.
{\em Limit theorems for stochastic processes}.
Springer Science \& Business Media, {\tt 288}, 2013.

\bibitem{jacquier2017vix}
{\sc A.~Jacquier, C.~Martini and A.~Muguruza}.
{\em On VIX Futures in the rough Bergomi model}.
Quantitative Finance, {\tt 18}(1): 45-61, 2018.

\bibitem{kallsen2010exponentially}
{\sc J.~Kallsen and J.~Muhle-Karbe}.
{\em Exponentially affine martingales, affine measure changes and exponential moments of affine processes}.
Stochastic Processes and their Applications, {\tt 120}: 163-181, 2010.

\bibitem{Kemna}
{\sc A.G.Z Kemna and A.C.F. Vorst}. 
{\em A pricing method for options based on average asset values}. 
Journal of Banking and Finance, {\tt 14}(1): 113-129, 1990.

\bibitem{Mandelbrot}
{\sc B. Mandelbrot and J.W. van Ness}. 
{\em Fractional Brownian motions, fractional noises and applications}.
SIAM Review, {\tt 10}(4): 422-437, 1968.

\bibitem{mccrickerd2017turbocharging}
{\sc R.~McCrickerd and M.~S. Pakkanen}.
{\em Turbocharging Monte Carlo pricing for the rough Bergomi model}.
Quantitative Finance, {\tt 18}(11): 1877-1886, 2018.

\bibitem{OuldAly}
{\sc S. M. Ould Aly}.
{\em Forward variance dynamics: Bergomi's model revisited}.
Applied Mathematical Finance, {\tt 21}(1): 87-104, 2013.

\bibitem{protter}
{\sc P.~Protter}.
{\em Stochastic Integration and Differential Equations.}
Springer-Verlag New York, 2014.

\end{thebibliography}
\end{document}